\documentclass[runningheads]{llncs}
\usepackage[T1]{fontenc}
\usepackage{graphicx}
\usepackage{hyperref}

\usepackage{color}

\usepackage{latexsym}
\usepackage{amssymb}
\usepackage{amsmath}
\usepackage{virginialake}
\vlnostructuressyntax

\usepackage{thmtools,thm-restate}

\usepackage{multicol}

\usepackage{ifthen}
\usepackage{colonequals}

\newcommand{\Paaras}[1]
{}%
\newcommand{\sonia}[1]
{}%
\newcommand{\todo}[1]
{}%

\newcommand{\newpaaras}[1]
{}
\newcommand{\AND}{\mathbin{\wedge}}
\newcommand{\OR}{\mathbin{\vee}}
\newcommand{\IMPLIES}{\mathbin{\rightarrow}} %\supset
\newcommand{\IMP}{\IMPLIES}

\newcommand{\DIA}[1][]{\Diamond_{#1}}
\newcommand{\BOX}[1][]{\Box_{#1}}
\newcommand{\bottom}{\mathord{\bot}}
\newcommand{\BOT}{\mathord{\bottom}}
\newcommand{\TOP}{\mathord{\top}}

\newcommand{\BNFequals}{\coloncolonequals}
\newcommand{\proves}[2]{#1 \vdash #2}
\newcommand{\nec}{\mathsf{nec}}

\newcommand{\logic}[1]{\mathsf{#1}}
\newcommand{\LP}{\logic{LP}}
\newcommand{\sfour}{\logic{S4}}
\newcommand{\sfive}{\logic{S5}}
\newcommand{\CK}{\logic{CK}}
\newcommand{\IK}{\logic{IK}}
\newcommand{\ISfour}{\logic{IS4}}

\newcommand{\IKfour}{\logic{\IK4}}
\newcommand{\IKt}{\logic{IKt}}
\newcommand{\prop}{\mathsf{Prop}}
\newcommand{\IPL}{\logic{IPL}}
\newcommand{\kax}[1][]{\mathsf{{k}_{#1}}}
\newcommand{\tax}{\mathsf{t}}
\newcommand{\taxb}{\mathsf{t}_{\BOX}}
\newcommand{\taxd}{\mathsf{t}_{\DIA}}
\newcommand{\fax}{\mathsf{4}}
\newcommand{\faxb}{\mathsf{4}_{\BOX}}
\newcommand{\faxd}{\mathsf{4}_{\DIA}}
\newcommand{\langmod}{\mathcal{L}_{\BOX}}
\newcommand{\langann}{\mathcal{L}_{\BOX}^{\mathsf{ann}}}
\newcommand{\langjust}{\mathcal{L}_{\mathsf{J}}}
\newcommand{\J}{\logic{J}}
\newcommand{\Log}{\logic{L}}
\newcommand{\JL}{\logic{JL}}

\newcommand{\nL}{\logic{nL}}
\newcommand{\nIK}{\logic{nIK}}
\newcommand{\nIKt}{\logic{nIKt}}
\newcommand{\nIKfour}{\logic{nIK4}}
\newcommand{\nISfour}{\logic{nIS4}}
\newcommand{\JIL}{\logic{JIL}}
\newcommand{\JIK}{\logic{JIK}}
\newcommand{\JIKfour}{\logic{JIK4}}
\newcommand{\JIKt}{\logic{JIKt}}
\newcommand{\JISfour}{\logic{JIS4}}
\newcommand{\ILP}{\logic{ILP}}
\newcommand{\jkax}[1]{\mathsf{{jk}_{#1}}}
\newcommand{\jtaxb}{\mathsf{jt}_{\BOX}}
\newcommand{\jtaxd}{\mathsf{jt}_{\DIA}}
\newcommand{\jfaxb}{\mathsf{j4}_{\BOX}}
\newcommand{\jfaxd}{\mathsf{j4}_{\DIA}}
\newcommand{\jsumax}{\mathsf{j}{\jsum{}{}}}
\newcommand{\juniax}{\mathsf{j}{\suni{}{}}}
\newcommand{\HA}{\logic{HA}}
\newcommand{\PA}{\logic{PA}}

\newcommand{\just}[2]{#1 {:} #2}
\newcommand{\sat}[2]{#1 {:} #2}
\newcommand{\prfterms}{\mathsf{PrfTm}}
\newcommand{\satterms}{\mathsf{SatTm}}
\newcommand{\prfvars}{\mathsf{PrfVar}}
\newcommand{\satvars}{\mathsf{SatVar}}
\newcommand{\pvar}[1][]{x_{#1}}
\newcommand{\svar}[1][]{\alpha_{#1}}
\newcommand{\respvar}[1][]{y_{#1}}
\newcommand{\ressvar}[1][]{\beta_{#1}}
\newcommand{\pconst}[1][]{c_{#1}}
\newcommand{\prfconst}{\mathsf{PrfConst}}
\newcommand{\jsum}[2]{#1 \mathbin{+} #2}
\newcommand{\jappl}[2]{#1 \mathbin{\cdot} #2}
\newcommand{\jbang}[1]{{!}#1}
\newcommand{\suni}[2]{#1 \mathbin{\sqcup} #2}
\newcommand{\sappl}[2]{#1 \mathbin{\star} #2}
\newcommand{\jupdt}[2]{#1 \mathbin{\triangleright} #2} %\rightharpoonup
\newcommand{\can}{\mathsf{can}}

\newcommand{\fseq}{\Gamma}
\newcommand{\lseq}{\Lambda}
\newcommand{\oseq}{\Pi}
\newcommand{\dseq}{\Delta}
\newcommand{\inp}[1]{{#1}^{\bullet}}
\newcommand{\out}[1]{{#1}^{\circ}}
\newcommand{\br}[2][]{\left[ #2 \right]_{#1}}
\newcommand{\diabr}[2][]{\langle #2 \rangle_{#1}}
\newcommand{\emptyseq}{\varnothing}
\newcommand{\fm}[1]{\text{\normalfont form} (#1)}
\newcommand{\defseq}{\Gamma}
\newcommand{\prunesymb}{\downarrow}
\makeatletter
\newcommand{\seq}[2][\defseq]{
    \def\cxs@{#2}
    #1 \{\ifx\cxs@\empty\mkern 3mu\else #2\fi\}
    }
\newcommand{\seqprune}[2][\defseq]{
    \def\cxs@{#2}
    #1^{\prunesymb} \{\ifx\cxs@\empty\mkern 3mu\else #2\fi\}}
\makeatother
\newcommand{\depth}[1]{\text{\normalfont depth}(#1)}

\newcommand{\id}{\mathsf{id}}
\newcommand{\botrule}{\inp{\BOT}}
\newcommand{\andl}{\inp{\AND}}
\newcommand{\andr}{\out{\AND}}
\newcommand{\orl}{\inp{\OR}}
\newcommand{\orr}[1][]{\out{\OR}_{#1}}
\newcommand{\impl}{\inp{\IMPLIES}}
\newcommand{\constrimpl}{\inp{\IMPLIES}_\mathsf{s}}
\newcommand{\impr}{\out{\IMPLIES}}
\newcommand{\boxl}{\inp{\BOX}}
\newcommand{\boxldia}{\inp{\BOX}_{\diabr{\cdot}}}
\newcommand{\boxlbox}{\inp{\BOX}_{\br{\cdot}}}
\newcommand{\boxr}{\out{\BOX}}
\newcommand{\dial}{\inp{\DIA}}
\newcommand{\diar}{\out{\DIA}}
\newcommand{\cont}{\inp{\mathsf{c}}}

\newcommand{\tr}{\out{\tax}}
\newcommand{\tl}{\inp{\tax}}
\newcommand{\fourl}{\inp{\fax}}
\newcommand{\fourldia}{\inp{\fax}_{\diabr{\cdot}}}
\newcommand{\fourlbox}{\inp{\fax}_{\br{\cdot}}}
\newcommand{\fourr}{\out{\fax}}
\newcommand{\updtrle}{\jupdt{}{}}
\newcommand{\rle}{\mathsf{rule}}

\newcommand{\set}[2][]{%
	\ifthenelse{\equal{#1}{}}%
	{\{#2\}}%
	{\{#1 \mid #2\}}%
}
\newcommand{\subst}[2][\sub]{#2 #1}
\newcommand{\sub}{\sigma}
\newcommand{\function}[3]{#1 : #2 \rightarrow #3}

\newcommand{\forget}[2][\forgetmap]{{#2}^{#1}}
\newcommand{\real}[2][r]{{#2}^{\rfunc{#1}}} %Refer to macro below
\newcommand{\rfunc}[1]{#1}

\newcommand{\nat}{\mathbb{N}}
\newcommand{\restr}[2][r]{#1 |_{#2}}
\newcommand{\UNI}{\cup}

\newcommand{\INT}{\cap}
\newcommand{\EMPTY}{\emptyset}
\newcommand{\comp}{\mathbin{\circ}}
\newcommand{\dom}[1]{\text{dom}(#1)}
\newcommand{\sset}{\subseteq}
\newcommand{\var}[1]{\text{var}(#1)}
\newcommand{\negvar}[1]{{\text{negvar}(#1)}}
\newcommand{\ann}[1]{\text{ann}(#1)}
\newcommand{\Idnty}{\mathsf{Id}}

\vlnosmallleftlabels

\newcommand{\vlstr}[3]{\vltr{#1}{#2}{\vlhy{}}{#3}{\vlhy{}}}
\newcommand{\vlhtr}[2]{\vlstr{#1}{#2}{\vlhy{\hskip1.5em}}}
\newcommand{\vlderivationauxnc}[1]{#1{\box\derboxone}\vlderivationterm}
\newcommand{\vlderivationnc}{\vlderivationinit\vlderivationauxnc}

\makeatletter
\let\c@proposition\c@theorem
\let\c@corollary\c@theorem
\let\c@lemma\c@theorem
\let\c@definition\c@theorem
\let\c@example\c@theorem
\let\c@remark\c@theorem
\makeatother

\begin{document}
	\title{Justification Logic for Intuitionistic Modal Logic (Extended Technical Report)}
	\titlerunning{Justification Logic for Intuitionistic Modal Logic}
	\author{Sonia Marin \and
		Paaras Padhiar}
	\authorrunning{S. Marin and P. Padhiar}
	\institute{School of Computer Science, University of Birmingham, Birmingham, United~Kingdom \\
		\email{s.marin@bham.ac.uk}\\
		\email{pxp367@student.bham.ac.uk}}
	\maketitle              % typeset the header of the contribution
	\begin{abstract}
		%The abstract should briefly summarize the contents of the paper in 150--250 words.
		Justification logics are an explication of modal logic; boxes are replaced with proof terms formally through realisation theorems.
This can be achieved syntactically using a cut-free proof system e.g. using sequent, hypersequent or nested sequent calculi.
In constructive modal logic, boxes and diamonds are decoupled and not De Morgan dual.
Kuznets, Marin and Stra{\ss}burger provide a justification counterpart to constructive modal logic~$\CK$ and some extensions by making diamonds explicit by introducing new terms called satisfiers. 
We continue the line of work to provide a justification counterpart to Fischer Servi’s intuitionistic modal logic~$\IK$ and its extensions with the~$\tax$ and~$\fax$ axioms. 
We: extend the syntax of proof terms to accommodate the additional axioms of intuitionistic modal logic; provide an axiomatisation of these justification logics; provide a syntactic realisation procedure using a cut-free nested sequent system for intuitionistic modal logic introduced by Stra{\ss}burger which can be further adapted to other nested sequent systems.
	\end{abstract}

	\section{Introduction}\label{sec:intro}
	Justification logics are a family of logics which refine modal logics by replacing modal operators with explicit \emph{proof terms}.
In a similar way to how a modal formula~$\BOX A$ can be read as \emph{$A$ is provable} in provability logics,  it can be given an explicit justification counterpart~$\just{t}{A}$ for some proof term~$t$, to be read as \emph{there exists a proof $t$ of $A$}.

The first such justification logic is the Logic of Proofs,~$\LP$, introduced by Artemov~\cite{artemov_operational_1995,artemov_explicit_2001} which is an explicitation of the modal logic~$\sfour$.
Artemov showed that this logic, which can be seen as provability semantics to intuitionistic logic~($\IPL$) via the G\"odel translation of $\IPL$ into $\sfour$, enjoys an arithmetic interpretation into Peano Arithmetic~($\PA$).
The formal connection between~$\LP$ and~$\sfour$ is made through a \emph{realisation theorem}, that is by translating each theorem of~$\sfour$ into a corresponding theorem of~$\LP$ by \emph{realising} every modalities with proof terms.
The first realisation theorem was proved syntactically by Artemov using a cut-free Gentzen sequent calculus. 
Fitting~\cite{fitting_logic_2005} later provided a semantic method for realisation.

$\LP$ was also further generalised, and realisation theorems were obtained for instance for the modal~$\sfive$ cube and the family of Geach logics (see~\cite{artemov_justification_2024} for a survey).
In parallel, Artemov's method for realisation can be adapted to modal logics if they have a cut-free Gentzen sequent calculus (see~\cite{kuznets_complexity_2008}).
However, this syntactic method was also expanded to more exotic cut-free proof systems: based on the formalism of hypersequents~\cite{artemov_logic_1999} and of nested sequents~\cite{brunnler_syntactic_2010,goetschi_realization_2012,borg_realization_2015}.

Since the early days of justification logic, there has been an interest in their intuitionistic variants.
Originally as a way to unify the treatment of terms between $\LP$ and the $\lambda$-calculus~\cite{artemov_unified_2002}.
Another line of research that has brought justification logic into the realm of intuitionistic modal logics, is the investigation of the provability logic of Heyting Arithmetic ($\HA$). 
Some propositions for an intuitionistic version of $\LP$ have been put forward in~\cite{artemov_basic_2007,dashkov_arithmetical_2011}, as a building block towards a justification logic for $\HA$.
Other intuitionistic variants of justification logics have been considered as a basis for type systems for computations that have access to some of their execution~\cite{steren_intuitionistic_2014}.
The semantics and the proof theory of some of these logics have also been partially studied~\cite{marti_intutionistic_2016,marti_internalized_2018,hill_analytic_2019}; however none of these work provide counterparts to an intuitionistic modal logic which contains both the box and the diamond modalities.

The first justification logic for an intuitionistic modal logic which makes the diamond explicit was postulated in~\cite{kuznets_justification_2021} as a justification counterpart to \emph{constructive modal logic} $\CK$~\cite{prawitz_natural_1965,bierman_intuitionistic_2000,bellin_extended_2001} via a syntactic realisation procedure.
As~$\BOX$ is replaced by a proof term~$t$; so is the diamond operator~$\DIA$ replaced with a \emph{satisfier} term~$\mu$.

We expand this line of work to provide a justification counterpart to the intuitionistic variant of modal logic $\IK$, as defined originally by~\cite{fischer_servi_modal_1977,fischer_servi_semantics_1980,fischer_servi_axiomatizations_1984,plotkin_framework_1986}.
Interestingly, intuitionistic modal logic $\IK$ is remarkably more expressive than constructive modal logic $\CK$ as it not only provides more validities about the $\DIA$ operator (in particular, that it is normal), but it also permits the derivation of more $\DIA$-free theorems~\cite{das_intuitionistic_2023}.
One such theorem is the G\"odel-Gentzen (double-negation) translation, which makes for example the relationship of $\IK$ with classical $\logic K$ more akin to the relationship of $\HA$ with $\PA$ than could be obtained with $\CK$, and is also more amenable to functional interpretations.
 
We present an axiomatisation of our justification counterpart to~$\IK$ and extensions with the $\tax$ and $\fax$ axioms and establish a formal connection through a syntactic realisation procedure using a cut-free system for~$\IK$ and extensions~\cite{strasburger_cut_2013} using nested sequents~\cite{bull_cut_1992,kashima_cut-free_1994,brunnler_deep_2006,brunnler_deep_2009,poggiolesi_method_2009,poggiolesi_gentzen_2011}.

    \section{Modal and Justification Logics}\label{sec:logics}
    \subsection{Intuitionistic Modal Logic}
	\emph{Formulas} of intuitionistic modal logic, $\langmod$, are given by the grammar:
$$A \BNFequals \bottom \mid p \mid (A \AND A) \mid (A \OR A) \mid (A \IMPLIES A) \mid \BOX A \mid \DIA A$$
where~$p$ ranges over a countable set of propositional atoms~$\prop$. 
We will drop the outermost brackets when it is not required. 
We will assume~$\AND$ and~$\OR$ are associative, $\AND$ is commutative, and~$\IMPLIES$ is right associative.

\begin{figure}[t]
\fbox{
    \begin{minipage}{.95\textwidth}
	\centering
    $
    \begin{array}{r@{\quad:\quad}l}
	\kax[1] & \BOX(A \IMPLIES B) \IMPLIES (\BOX A \IMPLIES \BOX B) 
    \\
	\kax[2] & \BOX(A \IMPLIES B) \IMPLIES (\DIA A \IMPLIES \DIA B) 
    \\
    \kax[3] &  \DIA(A \OR B) \IMPLIES (\DIA A \OR \DIA B) 
    \\
	\kax[4] & (\DIA A \IMPLIES \BOX B ) \IMPLIES \BOX (A \IMPLIES B) 
    \\
	\kax[5] &  \DIA \bottom \IMPLIES \bottom
    \end{array}
    \qquad
    \begin{array}{r@{\quad:\quad}l}
    \taxb & \BOX A \IMPLIES A
    \\
    \taxd & A \IMPLIES \DIA A 
    \\
    \faxb & \BOX A \IMPLIES \BOX \BOX A
    \\
    \faxd & \DIA \DIA A \IMPLIES \DIA A
    \\
    \end{array}
    $
    \qquad
    $\vlinf{\nec}{}{\proves{}{\BOX A}}{\proves{}{A}}$
    \end{minipage}
    }
    \caption{Intuitionistic modal axioms}
    \label{fig:mod-axioms}
\end{figure}

On Fig.~\ref{fig:mod-axioms} we list the modal axioms that we consider in this work. 
\emph{Constructive} modal logic $\CK$~\cite{bellin_extended_2001} is an extension of intuitionistic propositional logic~$\IPL$ with axioms~$\kax[1]$ and~$\kax[2]$ and the necessitation inference.
\emph{Intuitionistic} modal logic $\IK$~\cite{fischer_servi_modal_1977,fischer_servi_semantics_1980,fischer_servi_axiomatizations_1984,plotkin_framework_1986} is an extension of~$\CK$ with axioms~$\kax[3]$, $\kax[4]$ and~$\kax[5]$.
We will also study extensions of $\IK$: with axioms $\taxb$ and $\taxd$ to logic $\IKt$, with axioms $\faxb$ and $\faxd$ to logic $\IKfour$, and finally to logic~$\ISfour$ with all four axioms axioms.

    \subsection{Justification Logic}
    \emph{Formulas} of intuitionistic justification logic,~$\langjust$, are given by the grammar:
$$A \BNFequals \BOT \mid p \mid (A \AND A) \mid (A \OR A) \mid (A \IMPLIES A) \mid \just{t}{A} \mid \sat{\mu}{A}$$
where~$p$ ranges over a countable set propositional atoms~$\prop$, $t$~ranges over \emph{proof terms}~$\prfterms$ and $\mu$~ranges over \emph{satisfiers}~$\satterms$.

\emph{Proof terms}~$t,s,\dots$ and \emph{satisfiers}~$\mu,\nu,\dots$ are generated as follows:
\begin{align*}
	t &\BNFequals \pvar,\respvar \mid \pconst \mid (\jsum{t}{t}) \mid (\jappl{t}{t}) \mid (\jupdt{\mu}{t}) \mid \jbang{t} \\
	\mu &\BNFequals \svar,\ressvar \mid (\suni{\mu}{\mu}) \mid (\sappl{t}{\mu})
\end{align*}
where $\pvar,\respvar$ range over proof variables~$\prfvars$, $\svar,\ressvar$ over satisfier variables $\satvars$, and $c$ over proof constants $\prfconst$.
A term is called \emph{ground} if it contains neither proof nor satisfier variables.
For a formula $A$, we write $\var{A}$ for the set $\set[\pvar \in \prfvars]{\text{$\pvar$ occurs in A}}\UNI \set[\svar \in \satvars]{\text{$\svar$ occurs in A}}$.

\begin{figure}[t]
\fbox{
    \begin{minipage}{.95\textwidth}
    \centering
    $
    \begin{array}{r@{\ :\ }l}
	\jkax{1} & \just{s}{(A \IMPLIES B)} \IMPLIES (\just{t}{A} \IMPLIES \just{\jappl{s}{t}}{B})
	\\
	\jkax{2} & \just{s}{(A \IMPLIES B)} \IMPLIES (\sat{\mu}{A} \IMPLIES \just{\sappl{s}{\mu}}{B})
	\\
	\jkax{3} & \sat{\mu}{(A \OR B)} \IMPLIES (\sat{\mu}{A} \OR \sat{\mu}{B})
	\\
	\jkax{4} & (\sat{\mu}{A} \IMPLIES \just{t}{B}) \IMPLIES \just{\jupdt{\mu}{t}}{(A \IMPLIES B)}
	\\
	\jkax{5} & \sat{\mu}{\BOT} \IMPLIES \BOT
    \end{array}
\quad
\begin{array}{r@{\ :\ }l@{\qquad}r@{\ :\ }l}
	\jsumax & \just{s}{A} \IMPLIES \just{(\jsum{s}{t})}{A} 
    &
    \jtaxb & \just{t}{A} \IMPLIES A
    \\
	\jsumax & \just{t}{A} \IMPLIES \just{(\jsum{s}{t})}{A} 
    &
    \jtaxd & A \IMPLIES \sat{\mu}{A}
    \\
	\juniax & \sat{\mu}{A} \IMPLIES \sat{(\suni{\mu}{\nu})}{A} 
    &
    \jfaxb & \just{t}{A} \IMPLIES \just{\jbang{t}}{\just{t}{A}}
    \\
	\juniax & \sat{\nu}{A} \IMPLIES \sat{(\suni{\mu}{\nu})}{A} 
    &
    \jfaxd & \sat{\mu}{\sat{\nu}{A}} \IMPLIES \sat{\nu}{A}
\end{array}
$
\end{minipage}
}
    \caption{Intuitionistic justification axioms}
    \label{fig:justif-axioms}
\end{figure}

%Fix an axiomatisation of $\IPL$.
%
We define the \emph{justification counterpart of modal logic $\IK$}, which we call $\JIK$, as the extension of $\IPL$ with axioms $\jkax{1}-\jkax{5}$, $\jsumax$ and $\juniax$ on Fig.~\ref{fig:justif-axioms},
together with the \emph{constant axiom necessitation rule}:
$$
\vlderivationnc{\vlin{\can}{}{\just{\pconst[n]}{\just{\dots}{\just{\pconst[1]}{A}}}}{\vlhy{}}}
$$
where $\pconst[1], \dots, \pconst[n] \in \prfconst$ and $A$ is any axiom instance in Fig.~\ref{fig:justif-axioms}.

We furthermore define justification counterparts to the extensions of $\IK$ introduced  previously using the additional axioms on Fig.~\ref{fig:justif-axioms}:
$\JIKt = \JIK + \jtaxb + \jtaxd$, $\JIKfour = \JIK + \jfaxb + \jfaxd$ and $\JISfour = \JIK + \jtaxb + \jtaxd + \jfaxb + \jfaxd$.

%\Paaras{Need a note on operators of terms and satisfiers}
%\todo{decide whether this stays here as the axioms give the meaning to the operators or earlier when introducing the grammar}

The operations \emph{proof sum}~$\jsum{}{}$, \emph{application}~$\jappl{}{}$ and \emph{proof checker}~$\jbang{}$ are the standard justification operations relating to proof manipulations.

The operations of \emph{propagation}~$\sappl{}{}$ and \emph{disjoint union}~$\suni{}{}$ were introduced in~\cite{kuznets_justification_2021} to consolidate the intuition that in $\sat{\mu}{A}$, $\mu$ is some \emph{model} of $A$.
The operation~$\sappl{}{}$ can be seen as a combination of local and global reasoning, e.g. suppose $A$ is a local fact, i.e. we have $\sat{\mu}{A}$, we can use the proof $t$ of $A \IMPLIES A \OR B$ to locally reason $\sat{\sappl{t}{\mu}}{A \OR B}$.
The operation $\suni{}{}$ is akin to that of a disjoint union of models.

We introduce the operation \emph{local update}~$\jupdt{}{}$ which carries the intuition that if a local fact implies global knowledge, one can update the global knowledge with local data.

\subsection{From $\JL$ to $\Log$ -- and back}

In the rest of this paper, we will write $\Log$ for any logic in $\set{\IK, \IKt, \IKfour, \ISfour}$ and $\JL$ for the corresponding logic in $\set{\JIK, \JIKt, \JIKfour, \JISfour}$.
The main point of this work is to establish a syntactic correspondence between the two logics introduced in the previous sections.

\begin{definition}[Forgetful projection]
	The forgetful projection is a map~$\function{\forget{(\cdot)}}{\langjust}{\langmod}$ 
 %   \sonia{cannot be the same notation as the side distinguisher in nested sequent}
    inductively defined as follows:
    
    $
    \begin{array}{rl@{\hspace{5cm}}rl}
        \forget{\BOT} & \colonequals \BOT
        &  
        \forget{(\just{t}{A})} & \colonequals \BOX \forget{A}
        \\
        \forget{p} & \colonequals p
        & 
        \forget{(\sat{\mu}{A})} & \colonequals \DIA \forget{A}
        \\
        \forget{(A \ast B)} &
        \multicolumn{3}{l}{
        \colonequals (\forget{A} \ast \forget{B})$ where~$\ast \in \set{\AND, \OR, \IMPLIES}
        }
    \end{array}
    $
    
\end{definition}

\begin{theorem}
	%Let~$\Log \in \set{\IK, \IKt, \IKfour, \ISfour}$.
	%
	Let~$A \in \langjust$.
	%
	%Then
	If $\proves{\JL}{A}$ then $\proves{\Log}{\forget{A}}$.
\end{theorem}
\begin{proof}
	This follows from the fact that the forgetful projection on axioms of~$\JL$ and conclusions of the $\can$~rule are theorems of~$\Log$.
\end{proof}

\begin{definition}[Realisation]
    A realisation is a map $\function{\real{(\cdot)}}{\langmod}{\langjust}$ such that $\forget{(\real A)} = A$ for each $A \in \langmod$.
\end{definition}

\begin{theorem}\label{thm:realisation}
    %Let~$\Log \in \set{\IK, \IKt, \IKfour, \ISfour}$.
	%
	Let~$A \in \langmod$.
    If $\proves{\Log}{A}$, 
	then there exists a realisation $r$ such that
    $\proves{\JL}{\real A}$.
\end{theorem}

The goal of this paper is to prove this theorem. 
We first need to introduce the main technical ingredients of the proof. 

\subsection{Realisation over Annotations}

A subformula of A is \emph{positive} if its position in the formula tree of A is reached from the root by following the left branch of an $\IMP$ an even number of times; otherwise it is called \emph{negative}.

A realisation is \emph{normal} if when $t:B$ (resp.~$\mu:B$) is a negative subformula of $\real A$, then $t\in\prfvars$ (resp.~$\mu\in\satvars$) and occurs
in $A$ exactly once.

We follow and expand the methodology introduced in~\cite{goetschi_realization_2012} using nested sequents to prove a realisation theorem. 
Due to the complexity of the sequent structure, some extra machinery is required to ensure a careful bookkeeping during the proof manipulations.

\begin{definition}[Annotated formula]
	\emph{Annotated formulas} are built as in $\langmod$ but using annotated boxes $\BOX[n]$ and diamonds $\DIA[n]$ for $n \in \nat$.

    \noindent 
    For an annotated formula $A$, define~$\ann{A} \colonequals \set[n \in \nat]{\text{$\BOX[n]$ or $\DIA[n]$ occurs in $A$}}$.

\end{definition}

\begin{definition}[Properly annotated]
	A formula $A$ is \emph{properly annotated} if:
	\begin{itemize}
		\item modalities are annotated by pairwise distinct indexes;
%        \sonia{is it the same as saying that no annotation $n$ occurs twice?}\Paaras{yes}
        \item a $\BOX$ (resp.~$\DIA$) is indexed by $n$ iff $n = 0,1 \mod 4$ (resp.~$n = 2,3 \mod 4$);
		\item a modality indexed by $n$ occurs in a positive (resp.~negative) position iff $n$ is even (resp.~odd).
%        \sonia{should these be iff?}\Paaras{yes}
		%\item $\DIA[n]$ occurs in a positive position if $n$ is even and it occurs in a negative position if $n$ is odd.
	\end{itemize}
	The set of properly annotated formulas is denoted $\langann$.
\end{definition}

For $\Log \in \set{\IK, \IKt, \IKfour, \ISfour}$, we will write
$\proves{\Log}{A}$ for $A \in \langann$,
when the unannotated version of $A$ is a theorem of $\Log$ from now on.

We assume that we have fixed an enumeration of proof variables $\prfvars = \set{\pvar[1], \pvar[2], \pvar[3], \dots}$ and satisfier variables $\satvars = \set{\svar[1], \svar[2], \svar[3], \dots}$.
Similarly, we enumerate the reserved proof variables $\set{\respvar[1], \respvar[2], \respvar[3], \dots}$ and the satisfier variables $\set{\ressvar[1], \ressvar[2], \ressvar[3], \dots}$

\begin{definition}[Realisation on annotations]
	A \emph{realisation on annotations} is a partial function
	$$\function{r}{\nat}{\prfterms \UNI \satterms}$$
	where~$r(4n)\in \prfterms$, $r(4n+1)=\pvar[n]$, $r(4n+2)\in \satterms$, $r(4n+3)=\svar[n]$ when it is defined.
	Denote the domain of~$r$ as~$\dom{r} \colonequals \set[n \in \nat]{\text{$r(n)$ is defined}}$.
    \end{definition}

    A realisation $\real{A}$ for a formula $A\in\langmod$ can be obtained by first annotating properly all modalities occurring in it and then replacing by the term $r(n)$ given by the realisation on annotations.
	
	A realisation $r$ can be extended to an annotated formula~$A$ such that $\dom{r} \sset \ann{A}$ inductively:
	\begin{itemize}
		\item $\real{\BOT} \colonequals \BOT$
		\item $\real{p} \colonequals p$
		\item $\real{(A \ast B)} \colonequals (\real{A} \ast \real{B})$ where~$\ast \in \set{\AND, \OR, \IMPLIES}$
		\item $\real{(\BOX[n]A)} \colonequals \just{r(n)}{\real{A}}$
		\item $\real{(\DIA[n]A)} \colonequals \sat{r(n)}{\real{A}}$ with the \emph{satisfier self-referentiality restriction} that: if $n = 4k+3$, $\svar[k]$ does not occur in $\real{A}$
	\end{itemize}

	Given~$\ann{A} \sset \dom{r}$, denote~$\restr{A}$ the restriction of the partial function~$r$ onto the domain~$\ann{A}$.

Note that, when a properly annotated formula $A\in\langann$ is realised through this algorithm, it fixes automatically some of the variables that occur in $\real A$. 
Namely,
$\set[{\pvar[n]} \in \prfvars]{\text{$\BOX[4n+1]$ occurs in $A$}}
\UNI
\set[{\svar[n]} \in \satvars]{\text{$\DIA[4n+3]$ occurs in $A$}}
$
designates a set of fixed variables,
which we will call
$\negvar{A}$.

	\section{Nested Sequent Calculus}\label{sec:nested}
    
A nested sequent~$\fseq$ is comprised of two distinct parts: an LHS-sequent~$\lseq$ and an RHS-sequent~$\oseq$.
Formally, these sequents are unordered multisets
generated from the following grammar:
$$
\lseq \BNFequals \emptyseq \mid \inp{A} \mid \diabr{\lseq} \mid \lseq, \lseq 
\qquad 
\oseq \BNFequals \out{A} \mid \br{\fseq}
\qquad
\fseq \BNFequals \lseq, \oseq %\mid \oseq \qquad
%\qquad
$$
%\sonia{what about: $\fseq \BNFequals \lseq, \out{A} \mid \lseq, \br{\fseq}$ and rid of $\Pi$?}
%\sonia{I believe $\emptyseq$ is enough, no need for $\diabr{}$}\Paaras{Agreed}
where~$A$ ranges over~$\langmod$.
%\sonia{needs something about the black and white notations}
%
The dots can be seen as marking the polarity of a formula in a sequent: $\bullet$ represents the antecedent; $\circ$ represents the succedent of a sequents -- hence, the restriction of only one $\circ$-formula in a sequent $\oseq$ is a rendering of Gentzen's intuitionistic restriction~\cite{strasburger_cut_2013}.

The formula interpretation of nested sequents is defined as follows:
$$
\begin{array}{r@{\ \colonequals\ }l@{\qquad}r@{\ \colonequals\ }l}
    \fm{\emptyseq} & \TOP 
    \\
	\fm{\inp{A}} & A 
    &
    \fm{\out{A}} & A
    \\
	\fm{\diabr{\lseq}} & \DIA \fm{\lseq} 
    &
    \fm{\br{\fseq}} & \BOX \fm{\fseq}
    \\
	\fm{\lseq_1, \lseq_2} & 
    \fm{\lseq_1} \AND \fm{\lseq_2} 
    &
    \fm{\lseq, \oseq} & 
    \fm{\lseq} \IMPLIES \fm{\oseq}
\end{array}
$$

\begin{remark}
    An immediate consequence of the definition is that for any LHS sequent~$\lseq$ in
    a sequent $\fseq$, $\fm{\lseq}$ is a negative subformula of $\fm{\fseq}$ and for a RHS sequent $\oseq$, $\fm{\oseq}$ is a positive subformula of $\fm{\fseq}$.
\end{remark}

\begin{definition}[Contexts]
    A \emph{context} $\seq[\dseq]{}$ is like a sequent but contains a hole wherever a formula may otherwise occur.
    It is an \emph{input context} $\seq{}$ if the hole should be filled with a LHS sequent to give a %full 
    sequent.
    It is an \emph{output context} $\seq[\lseq]{}$ if the hole should be filled with a %RHS
    sequent to give a %full 
    sequent.
    %
    %An LHS context is a sequent with a hole which should be filled with an LHS sequent to give an LHS sequent.
    %
    Formally, an input and output context can be inductively defined as follows:
    $$
        \seq[\lseq]{} \BNFequals \lseq, \seq[]{} \mid \lseq, \br{\seq[\lseq]{}}
        \quad
        \seq{} \BNFequals \seq[\Omega]{}, \oseq \mid \lseq, \br{\seq[\Omega]{}, \oseq}
    $$
    where $\seq[\Omega]{}$ is an LHS-context to be filled with a LHS sequent defined inductively:
    $$
    \seq[\Omega]{} \BNFequals \lseq, \seq[]{} \mid \lseq, \diabr{\lseq, \seq[\Omega]{}}
    $$
    The filling of context can then be inductively defined where if 
    $$\seq[\lseq]{} = \lseq_1, \seq[]{} \quad \seq[\Omega]{} = \lseq_1, \seq[]{}$$ 
    for some LHS sequent $\lseq_1$, then
    $$\seq[\lseq]{\fseq} \colonequals \lseq_1, \fseq \quad \seq[\Omega]{\lseq_2} \colonequals \lseq_1, \lseq_2$$ 
    for some sequent $\fseq$ and LHS sequent $\lseq_2$.
\end{definition}

\begin{definition}[Depth]
    The depth of a context is defined inductively:
    $$
    \begin{array}{r@{\ \colonequals\ }l@{\qquad}r@{\ \colonequals\ }l}
        \depth{\seq[]{}} & 0 
        \\
        \depth{\inp{A}, \seq{}} & \depth{\seq{}} 
        &
        \depth{\out{A}, \seq{}} & \depth{\seq{}} 
        \\
        \depth{\diabr{\seq[\lseq]{}}} & \depth{\seq[\lseq]{}} + 1
        &
        \depth{\br{\seq{}}} & \depth{\seq{}} + 1 
    \end{array}
    $$
\end{definition}

\begin{definition}
    For every input context $\seq{}$, its output pruning $\seqprune{}$ is an output context with the hole in the same position as $\seq{}$ but with the output formula in $\seq{}$ removed and brackets changed from $\br{\cdot}$ to $\diabr{\cdot}$ as necessary.
    
    Formally, we define the output pruning of a sequent~$\fseq$ inductively:
    \begin{itemize}
    	\item if $\fseq = \lseq, \oseq$, then $\fseq^\prunesymb = \lseq, \oseq^\prunesymb$;
    	\item if $\oseq = \out{A}$, then $\oseq^\prunesymb = \emptyseq$;
    	\item if $\oseq = \br{\fseq}$, then $\oseq^\prunesymb = \diabr{\fseq^\prunesymb}$.
    \end{itemize}
    and we define the output pruning of an input context $\seq{}$ inductively:
    \begin{itemize}
    	\item if $\seq{} = \seq[\Omega]{}, \oseq$, then $\seqprune{} = \seqprune[\Omega]{}, \oseq^\prunesymb$;
    	\item if $\seq{} = \lseq, \br{\seq[\Omega]{}, \oseq}$ then $\seqprune{} = \lseq, \br{\seqprune[\Omega]{}, \oseq^\prunesymb}$.
    	\item if $\seq[\Omega]{} = \lseq, \seq[]{}$, then  $\seqprune[\Omega]{} = \lseq, \seq[]{}$;
    	\item if $\seq[\Omega]{} = \lseq_1, \diabr{\lseq_2, \seq[\Omega']{}}$, then $\seqprune[\Omega]{} = \lseq_1, \br{\lseq_2, \seqprune[\Omega']{}}$
    \end{itemize}
    where $\lseq, \lseq_1, \lseq_2$ are LHS sequents, $\oseq$ is an RHS sequent, and $\seq[\Omega]{}, \seq[\Omega']{}$ are LHS contexts.
\end{definition}
\begin{example}
    Given $\seq{} = \lseq_1, \br{\diabr{\lseq_2}, \diabr{\lseq_3, \seq[]{}}, \out{C}}$, we get:
    $\seqprune{} = \lseq_1, \br{\diabr{\lseq_2}, \br{\lseq_3, \seq[]{}}}$.
\end{example}

The proof system~$\nIK$ consists of the rules given in Fig.~\ref{fig:nIK}.
The addition of the rules in Fig.~\ref{fig:nIS4} yields the proof systems~$\logic{nL} \in \set{\nIK, \nIKt, \nIKfour, \nISfour}$ by picking and mixing rules correspondong to axioms $\tax$ and $\fax$.

A proof of a nested sequent~$\fseq$ in~$\logic{nL}$ is constructed as trees from these rules with root~$\fseq$ and leaves closed by axiomatic rules, and we write $\proves{\logic{nL}}{\fseq}$.
%\sonia{that is only when all the premisses are closed}

\begin{figure}[t]
	\fbox{
    \begin{minipage}{.95\textwidth}
	\centering
	$\vlinf{\botrule}{}{\seq{\inp{\BOT}}}{} 
    \qquad 
    \vlinf{\id}{}{\seq[\lseq]{\inp{p}, \out{p}}}{}$
	\\[1ex]
	$\vlinf{\cont}{}{\seq{\lseq}}{\seq{\lseq, \lseq}}
    \qquad
    \vlinf{\andl}{}{\seq{\inp{A \AND B}}}{\seq{\inp{A}, \inp{B}}} 
    \qquad 
    \vliinf{\andr}{}{\seq[\lseq]{\out{A \AND B}}}{\seq[\lseq]{\out{A}}}{\seq{\out{B}}}$
	\\[1ex]
	$\vliinf{\orl}{}{\seq{\inp{A \OR B}}}{\seq{\inp{A}}}{\seq{\inp{B}}} 
    \qquad 
    \vlinf{\orr[1]}{}{\seq[\lseq]{\out{A \OR B}}}{\seq[\lseq]{\out{A}}} 
    \qquad 
    \vlinf{\orr[2]}{}{\seq[\lseq]{\out{A \OR B}}}{\seq[\lseq]{\out{B}}}$
	\\[1ex]
	$\vliinf{\impl}{}{\seq{\inp{A \IMPLIES B}}}{\seqprune{\out{A}}}{\seq{\inp{B}}} 
    \qquad 
    \vlinf{\impr}{}{\seq[\lseq]{\out{A \IMPLIES B}}}{\seq[\lseq]{\inp{A}, \out{B}}}$
	\\[1ex]
	$\vlinf{\boxlbox}{}{\seq[\lseq]{\inp{\BOX A}, \br{\fseq}}}{\seq[\lseq]{\br{\inp{A}, \fseq}}} 
    \qquad 
    \vlinf{\boxldia}{}{\seq{\inp{\BOX A}, \diabr{\lseq}}}{\seq{\diabr{\inp{A}, \lseq}}} 
    \qquad 
    \vlinf{\boxr}{}{\seq[\lseq]{\out{\BOX A}}}{\seq[\lseq]{\br{\out{A}}}}$
	\\[1ex]
	$\vlinf{\dial}{}{\seq{\inp{\DIA A}}}{\seq{\diabr{\inp{A}}}} 
    \qquad 
    \vlinf{\diar}{}{\seq[\lseq_1]{\out{\DIA A}, \diabr{\lseq_2}}}{\seq[\lseq_1]{\br{\out{A}, \lseq_2}}}$
    \end{minipage}
    }
	\caption{System $\nIK$}
	\label{fig:nIK}
	
	\bigskip
	\fbox{
    \begin{minipage}{.95\textwidth}
	\centering
	$\vlinf{\tl}{}{\seq{\inp{\BOX A}}}{\seq{\inp{A}}} 
    \qquad 
    \vlinf{\tr}{}{\seq[\lseq]{\out{\DIA A}}}{\seq[\lseq]{\out{A}}}$
	\\[1ex]
	$\vlinf{\fourlbox}{}{\seq[\lseq]{\inp{\BOX A}, \br{\fseq}}}{\seq[\lseq]{\br{\inp{\BOX A}, \fseq}}} 
    \qquad 
    \vlinf{\fourldia}{}{\seq{\inp{\BOX A}, \diabr{\lseq}}}{\seq{\diabr{\inp{\BOX A}, \lseq}}} 
    \qquad 
    \vlinf{\fourr}{}{\seq[\lseq_1]{\out{\DIA A}, \diabr{\lseq_2}}}{\seq[\lseq_1]{\br{\out{\DIA A}, \lseq_2}}}$
    \end{minipage}
    }
	\caption{Modal rules for $\tax$ and $\fax$ %\Paaras{Can Sonia box this up nicely please?}
    }
	\label{fig:nIS4}
\end{figure}

\begin{theorem}[Stra{\ss}burger~\cite{strasburger_cut_2013}]
\label{thm:soundcomp}
	Let~$\logic{L} \in \set{\IK, \IKt, \IKfour, \ISfour}$.
    For~$A \in \langmod$:
	$$\proves{\logic{L}}{A} \iff \proves{\logic{nL}}{\out{A}}$$
\end{theorem}
This is proved using a cut-elimination argument where the cut rule is of the shape
$$\vliinf{\mathsf{cut}}
{}
{\seq{\emptyseq}}
{\seqprune{\out{A}}}
{\seq{\inp{A}}}
$$
For reasons why we require use of a cut-free system, the reader is referred to~\cite[Chapter~8.8]{artemov_justification_2019}.

\begin{example} \label{ex:nested-pf}
This is an example of a proof in $\nIK$:
$$
\vlderivation{
    \vlin{\impr}{}{\out{((\BOX \BOT \IMP \BOT)\IMP \BOT)\IMP\BOX\BOT}}{
        \vlin{\boxr}{}{\inp{((\BOX \BOT \IMP \BOT)\IMP \BOT)},\out{\BOX\BOT}}{
            \vliin{\impl}{}{\inp{(\BOX \BOT \IMP \BOT)\IMP \BOT},\br{\out{\BOT}}}{
            \vlin{\impr}{}{\out{\BOX \BOT \IMP \BOT},\diabr{\emptyseq}}{
                \vlin{\boxl}{}{\inp{\BOX \BOT},\out{\BOT},\diabr{\emptyseq}}{
                    \vlin{\botrule}{}{\out{\BOT},\diabr{\inp{\BOT}}}{
                        \vlhy{}
                    }
                }
            }
            }{
            \vlin{\botrule}{}{\inp{\BOT},\br{\out{\BOT}}}{
                \vlhy{}
            }
            }
        }
    }
}
$$
Note that this is a theorem of $\IK$ but not of $\CK$~\cite{das_intuitionistic_2023}. 
We will expand this example later to illustrate how our system allows us to provide a realisation of it.
\end{example}

\begin{remark}
    The system given in~\cite{strasburger_cut_2013} doesn't use an explicit contraction rule as it is admissible due to the shape of the rules. 
    For the purposes of our work, we require this distinction to be made and decouple contraction from the rules similar to the rules in~\cite{marin_label-free_2014}.
\end{remark}
	
	\subsection{Annotated Sequents}
    \begin{definition}[Annotated sequents]
An \emph{annotated sequent/context} is a sequent/context in which only annotated formulas occur, and brackets $\br{\cdot}$ and $\diabr{\cdot}$ are indexed by some $n \in \nat$, written $\br[n]{\cdot}$ and $\diabr[n]{\cdot}$.
$$\fm{\diabr[n]{\lseq}} \colonequals \DIA[n] \fm{\lseq}
\qquad
\fm{\br[n]{\fseq}} \colonequals \BOX[n] \fm{\fseq}$$
\end{definition}

Most definitions for annotated formulas and for nested sequents transfer directly to annotated sequents. 
However, the concept of output pruning requires some careful handling as not only the brackets need to be possibly flipped, but their annotations need to be accordingly and appropriately changed as well.

The annotated rules are the same as the ones without annotations (but are given fully on Fig.~\ref{fig:nIKann} and~\ref{fig:nIS4ann} in Appendix).

The contraction rule applies in principle to the same sequent, but needs to distinguish annotations, that is, it should be written:
$$\vlinf{\cont}{}{\seq{\lseq_3}}{\seq{\lseq_1, \lseq_2}}$$
where $\lseq_1$, $\lseq_2$ and $\lseq_3$ represent the same unannotated sequent but do not share annotation indices.

\begin{example}
Let us illustrate the subtlety of the annotated output pruning on an instance of the $\impl$-rule.
    $$
    \vliinf{\impl}{}
    {
    \seq[\Omega_1]{\inp{A \IMPLIES B}},
    \seq[\Omega_2]{\out{C}}
    }
    {
    \seq[\Omega_1^{\prunesymb}]{\out{A}},
    \seq[\Omega_2^{\prunesymb}]{\emptyseq}
    }
    {
    \seq[\Omega_1]{\inp{B}},
    \seq[\Omega_2]{\out{C}}
    }
    $$
Suppose that $\seq[\Omega_1]{\inp{A\IMP B}} = \diabr[4k+3]{\lseq_1, \inp{A\IMP B}}$ and $\seq[\Omega_2]{\out{C}} = \br[4n]{\lseq_2,\out{C}}$.
Then, $\seq[\Omega_1^{\prunesymb}]{\out{A}}$ will be of the form $\br[4l]{\lseq_1, \out{A}}$ for some $l\in\nat$.
And, $\seq[\Omega_2^{\prunesymb}]{\emptyseq}$ of the form $\diabr[4m+3]{\lseq_2}$ for some $m\in\nat$.
\end{example}

    From now on, a nested sequent system $\logic{nL}$ will casually refer to the annotated version.
    We note that Theorem~\ref{thm:soundcomp} also holds for the annotated case -- for any unannotated derivation in the nested sequent system, properly annotate the endsequent and propagate the annotations upwards using the rules in Fig.~\ref{fig:nIKann} and~\ref{fig:nIS4ann}.

\begin{example} \label{ex:ann-pf}
Proof from Example~\ref{ex:nested-pf} in an annotated version:
$$
\vlderivation{
    \vlin{\impr}{}{\out{((\BOX[1] \BOT \IMP \BOT)\IMP \BOT)\IMP\BOX[0]\BOT}}{
        \vlin{\boxr}{}{\inp{((\BOX[1] \BOT \IMP \BOT)\IMP \BOT)},\out{\BOX[0]\BOT}}{
            \vliin{\impl}{}{\inp{(\BOX[1] \BOT \IMP \BOT)\IMP \BOT},\br[0]{\out{\BOT}}}{
            \vlin{\impr}{}{\out{\BOX[1] \BOT \IMP \BOT},\diabr[3]{\emptyseq}}{
                \vlin{\boxldia}{}{\inp{\BOX[1] \BOT},\out{\BOT},\diabr[3]{\emptyseq}}{
                    \vlin{\botrule}{}{\out{\BOT},\diabr[3]{\inp{\BOT}}}{
                        \vlhy{}
                    }
                }
            }
            }{
            \vlin{\botrule}{}{\inp{\BOT},\br[0]{\out{\BOT}}}{
                \vlhy{}
            }
            }
        }
    }
}
$$
\end{example}

\begin{remark}\label{rem:impl}
Each occurence of $\impl$ in a proof can be replaced by a macro rule composed of $\constrimpl$, $\updtrle$ and $\cont$ rules:
$$
\vliinf{\constrimpl}{}
        {
		\seq{\seq[\Omega]{\inp{A \IMPLIES B}}, \lseq, \oseq}
	}
	{
		\seq{\seq[\Omega^{\prunesymb}]{\out{A}}, \lseq}
	}
	{
		\seq{\seq[\Omega]{\inp{B}}, \oseq}
	}
    \qquad
    \vlinf{\updtrle}{}
	{
		\seq{\br[4n]{\lseq_1, \lseq_2, \oseq}}
	}
	{
		\seq{\diabr[4k+3]{\lseq_1}, \br[4n]{\lseq_2, \oseq}}	
	}
$$
where $\lseq, \lseq_1, \lseq_2$ are annotated LHS sequents, $\oseq$ is an annotated RHS sequent,
and the output pruning follows the requirements mentioned earlier. \newpaaras{Cite arXiv article}
\end{remark}

    \section{Realisation Theorem}\label{sec:real}

A realisation can be extended to an annotated sequent $\Sigma$ as $\real{\Sigma} \colonequals \real{\fm{\Sigma}}$.

The goal of this section is to prove the following theorem.

\begin{theorem}[Realisation Theorem]\label{thm:nestreal}
        Let $\Log \in \set{\IK, \IKt, \IKfour, \ISfour}$.
	Let~$\fseq$ be an annotated full sequent.
	If
	$$
	\proves{\nL}{\fseq}
	$$
	then there exists a realisation function~$r$ on~$\fseq$ such that
	$$
	\proves{\JL}{\real{\fseq}}
	$$
\end{theorem}

The high-level idea is to systematically scan through the proof of $\fseq$ top down. 
In the leaves, a basic realisation is given to any axiomatic sequent $\lseq,\oseq$ defined as:
\begin{multline*}
		r \colonequals \set[(4m, {\respvar[m]})]{\text{$4m$ is an annotation in~$\lseq,\oseq$}} 
		\\ \UNI \set[(4m+1, {\pvar[m]})]{\text{$4m+1$ is an annotation in~$\lseq,\oseq$}} 
		\\ \UNI \set[(4m+2, {\ressvar[m]})]{\text{$4m+2$ is an annotation in~$\lseq,\oseq$}} 
		\\ \UNI \set[(4m+3, {\svar[m]})]{\text{$4m+3$ is an annotation in~$\lseq,\oseq$}}
\end{multline*}
%%%%%%%%%%%%%%%%%%%%%%%%%%%%%%%
%%%%%%%%%%%%%%%%%%%%%%%%%%%%%%%
%%%%%%%%%%%%%%%%%%%%%%%%%%%%%%%
\begin{example} 
On the example annotated proof from Example~\ref{ex:ann-pf}, at each stage of the proof, the realisation procedure will yields a justification formula for each sequent in the derivation tree.
These formulas can be linked by derivations in the Hilbert system of the justification logics. 
The variables $\respvar[i]$ and $\ressvar[j]$ are step-by-step substituted top-down by larger and larger terms according to the $\JL$ axioms.
    $$
\vlderivation{
    \vlin{\impr}{}{((\just{\pvar[0]}{\BOT} \IMP \BOT)\IMP \BOT)\IMP\just{t(\jupdt{\svar[0]}{\respvar[0]})}{\BOT}}{
        \vlin{\boxr}{}{((\just{\pvar[0]}{\BOT} \IMP \BOT)\IMP \BOT)\IMP\just{t(\jupdt{\svar[0]}{\respvar[0]})}{\BOT}}{
            \vliin{\impl}{}{((\just{\pvar[0]}{\BOT} \IMP \BOT)\IMP \BOT)\IMP\just{t(\jupdt{\svar[0]}{\respvar[0]})}{\BOT}}{
            \vlin{\impr}{}{\sat{\svar[0]}{\TOP} \IMP (\just{\pvar[0]}{\BOT} \IMP \BOT)}{
                \vlin{\boxl}{}{(\sat{\svar[0]}{\TOP} \AND \sat{\pvar[0]}{\BOT} )\IMP\BOT}{
                    \vlin{\botrule}{}{\sat{\svar[0]}{\BOT}\IMP\BOT}{
                        \vlhy{}
                    }
                }
            }
            }{
            \vlin{\botrule}{}{\BOT\IMP\just{\respvar[0]}{\BOT}}{
                \vlhy{}
            }
            }
        }
    }
}
$$

    The proof term $t(\jupdt{\svar[k]}{\respvar[m]})$ is constructed from the steps given in Lemma~\ref{lem:rightrules}.
    Using the $\jkax{4}$ axiom
    $$
    \proves{\J}
    {
        (\sat{\svar[k]}{\TOP} \IMP \just{\respvar[m]}{\BOT})
        \IMP
        \just{\jupdt{\svar[k]}{\respvar[m]}}{(\TOP \IMP \BOT)}
    }
    $$
    and using the Lifting Lemma~\ref{lem:lifting} on the propositional theorem
    $$
    \proves{\J}
    {
        (\TOP \IMP \BOT)
        \IMP
        \BOT
    }
    $$
    there exists a proof term $t(\jupdt{\svar[k]}{\respvar[m]})$ such that
    $$
    \proves{\J}
    {
        \just{\jupdt{\svar[k]}{\respvar[m]}}{(\TOP \IMP \BOT)}
        \IMP
        \just{t(\jupdt{\svar[k]}{\respvar[m]})}{\BOT}
    }
    $$
\end{example}
%%%%%%%%%%%%%%%%%%%%%%%%%%%%%%%
%%%%%%%%%%%%%%%%%%%%%%%%%%%%%%%
%%%%%%%%%%%%%%%%%%%%%%%%%%%%%%%

\subsection{Lifting}

\begin{restatable}[Internalised neccessitation]{lemma}{lemnec}
\label{lem:nec}
	Let~$\J \in \set{\JIK, \JIKt, \JIKfour, \JISfour}$.
	Let~$A \in \langjust$.
	If~$\proves{\J}{A}$, then there exists a ground term~$t$ 
    such that
	$$\proves{\J}{\just{t}{A}}$$ 
\end{restatable}
\begin{proof}
	The proof is the same as the classical case. Example proofs can be found in~\cite{artemov_justification_2019,kuznets_logics_2019}.
    %\sonia{proof needs to be written here}
    \qed
\end{proof}

\begin{lemma}[Lifting Lemma]\label{lem:lifting}
	Let~$\J \in \set{\JIK, \JIKt, \JIKfour, \JISfour}$.
	Let $B_1, \dots, B_n, C, A \in \langjust$ for some~$n \in \nat$.
	%
	%Then if: 
        \begin{enumerate}
            \item If $\proves{\J}{B_1 \IMPLIES B_2 \IMPLIES \dots \IMPLIES B_n \IMPLIES A}$,
	then for proof terms~$s_1, \dots, s_n$, there exists a proof term~$t(s_1, \dots, s_n)$ such that
        $$\proves{\J}{\just{s_1}{B_1} \IMPLIES \dots \IMPLIES \just{s_n}{B_n} \IMPLIES \just{t(s_1, \dots, s_n)}{A}}$$
            \item If $\proves{\J}{B_1 \IMPLIES B_2 \IMPLIES \dots \IMPLIES B_n \IMPLIES C \IMPLIES A}$,
            then for proof terms~$s_1, \dots, s_n$ and satisfier~$\nu$, there exists a satisfier~$\mu(s_1, \dots, s_n, \nu)$ such that
            $$\proves{\J}{\just{s_1}{B_1} \IMPLIES \dots \IMPLIES \just{s_n}{B_n} \IMPLIES \sat{\nu}{C} \IMPLIES \sat{\mu(s_1, \dots, s_n, \nu)}{A}}$$
        \end{enumerate}
\end{lemma}
\begin{proof}
        A proof of 1.~can be found in~\cite{artemov_justification_2019,kuznets_logics_2019}.
	A proof of 2.~can be found in~\cite{kuznets_justification_2021}.
	The proof is not sensitive to the additional axioms and operators used here.
    \qed
\end{proof}

\subsection{Merging}

\begin{definition}[Substitution]
	A \emph{substitution} is a map~$\function{\sub}{\prfvars}{\prfterms}$ and~$\function{\sub}{\satvars}{\satterms}$. 
	Substitutions can then be inductively extended to~$\prfterms$, $\satterms$ and $\langjust$ as standard.
	We will write $\subst{\chi}$ for $\sub(\chi)$ for some $\chi \in \prfterms \UNI \satterms \UNI \langjust$.
	
	The \emph{domain} of~$\sub$ is the set
	$$\dom{\sub} \colonequals \set[\chi \in \prfvars \UNI \satvars]{\subst{\chi} \neq \chi}$$
\end{definition}

\begin{lemma}[Substitution Lemma]\label{lem:subst}
	Let~$\J \in \set{\JIK, \JIKt, \JIKfour, \JISfour}$.
	Let~$A \in \langjust$. 
	Let $\sub$ be a substitution. 
	If
	$\proves{\J}{A}$, then $\proves{\J}{\subst{A}}$.
\end{lemma}
\begin{proof}
	This is a routine proof which follows precisely from the fact that given any axiom instance~$A$, $\subst{A}$ is an axiom instance, and hence, for a conclusion of the $\can$~rule
    $\just{\pconst[n]}{\just{\dots}{\just{\pconst[1]}{A}}}$,
    the result of applying a substitution on it:
    $\subst{(\just{\pconst[n]}{\just{\dots}{\just{\pconst[1]}{A}}})} = \just{\pconst[n]}{\just{\dots}{\just{\pconst[1]}{(\subst{A})}}}
    $ is a conclusion of the $\can$~rule.
    \qed
\end{proof}

\begin{theorem}[Realisation merging]\label{thm:merging}
	Let~$\J \in \set{\JIK, \JIKt, \JIKfour, \JISfour}$.
	Let~$A \in \langann$.
	Let~$r_1$ and~$r_2$ be realisations on~$A$.
	Then there exists a realisation function~$r$ on~$A$ and a substitution~$\sub$ such that:
	\begin{enumerate}
		\item For every positive subformula~$X$ of~$A$, $\proves{\J}{\subst{\real[r_1]{X}} \IMPLIES \subst{\real[r_2]{X}} \IMPLIES \real{X}}$
		\item For every negative subformula~$X$ of~$A$, $\proves{\J}{\real{X} \IMPLIES \subst{\real[r_i]{X}}}$ where~$i \in \set{1, 2}$.
	\end{enumerate}
	where $\dom{\sub} \sset \negvar{A}$ and $\subst{x}$ contains no new variables for each~$x$.
\end{theorem}
\begin{proof}
	This is an adaptation of Fitting~\cite{fitting_realizations_2007,fitting_realizations_2009}.
    The original proof is in the setting of classical Logic of Proofs but does not make use of classical reasoning or the $\jtaxb$ and $\jfaxb$ axioms.
    The proof can be adapted to deal with satisfiers and has a similar treatment to proof terms, with $\suni{}{}$ playing for satisfiers the role played by $\jsum{}{}$ for proof terms.
    \qed
\end{proof}

These notions are similarly extended to annotated sequents.

\begin{corollary}\label{cor:nestedmerging}
	Let~$\J \in \set{\JIK, \JIKt, \JIKfour, \JISfour}$.
	Let~$\fseq$ be an annotated sequent and~$\lseq$ an annotated LHS sequent.
	Let~$r_1$ and~$r_2$ be realisation functions.
	Then there exists a realisation function~$r$ and a substitution~$\sub$ such that:
	\begin{enumerate}
		\item $\proves{\J}{\subst{\real[r_1]{\fseq}} \IMPLIES \subst{\real[r_2]{\fseq}} \IMPLIES \real{\fseq}}$ where~$\dom{\sub} \sset \negvar{\fseq}$ and $\subst{x}$ contains no new variables
		\item $\proves{\J}{\real{\lseq} \IMPLIES \subst{\real[r_i]{\lseq}}}$ where~$i \in \set{1, 2}$, $\dom{\sub} \sset \negvar{\fseq}$ and $\subst{x}$ contains no new variables.
	\end{enumerate}
\end{corollary}

    \subsection{Proving realisation}

Given an instance of a~$\rle \in \set{\id, \andr, \orr, \impr, \boxr, \diar, \boxlbox, \tr, \fourr, \fourlbox, \updtrle}$, 
for some $k\le n$,
it will be of the shape:
    $$
    \vliiinf{\rle}{}
    {
        \seq[\lseq]{\fseq_0}
    }
    {
        \seq[\lseq]{\fseq_1}
    }
    {
        \cdots
    }
    {
        \seq[\lseq]{\fseq_k}
    }
    $$

As these rules behave similarly to classical rules, following~\cite{goetschi_realization_2012} we can show:

\begin{restatable}{lemma}{lemrealwhite}
\label{lem:realwhite}
    %Let $\J \in \set{\JIK, \JIKt, \JIKfour, \JISfour}$.
    %
    For any realisations~$r_i$ on $\seq[\lseq]{\fseq_i}$ for $i\le k$, 
    there exist a realisation~$r$ on~$\seq[\lseq]{\fseq_0}$ 
    and substitutions $\sub_i$ %with $\dom{\sub_j} \sset \negvar{\fseq_j}$ for $j \in \set{1, \dots, i}$ 
    for $i\le k$, 
    such that
    $$
    \proves{\J}
    {
        \left(
        \subst[\sub_1]{\real[r_1]{\seq[\lseq]{\fseq_1}}}
        \AND
        \dots
        \AND
        \subst[\sub_k]{\real[r_k]{\seq[\lseq]{\fseq_k}}}
        \right)
        \IMPLIES
        \real{\seq[\lseq]{\fseq_0}}
    }
    $$
\end{restatable}
%In the case of when:
%\begin{itemize}
%    \item  $\rle \in \set{\id, \botrule}$, $i=0$;
%    \item $\rle \in \set{\tl, tr}$, $\J \in \set{\JIKt, \JISfour}$;
%    \item $\rle = \in \set{\fourlbox, \fourldia, \fourr}$, $\J \in \set{\JIKfour, \JISfour}$.
%\end{itemize}

\begin{proof}[Sketch]
We proceed by induction on $\seq[\lseq]{}$. 

The base case, $\lseq = \emptyseq$, relies on a bespoke proof for each specific $\rle$. These are similar to the classical case for $\id$, $\orr$, $\andr$ and $\impr$. 
Note the particular case of $\id$ where the empty conjunction reduces (see Lemma~\ref{lem:idrule} in Appendix)
and the particular case of $\andr$ whose branching requires the merging property, Theorem~\ref{thm:merging} (see Lemma~\ref{lem:andrrule} in Appendix). 
However, rules who involve a $\DIA$-formula (or $\diabr{\cdot}$-brackets) require a novel treatment as they introduce satisfier terms, rather than the classical treatment through duality with $\BOX$ introducing only proof terms (see Proposition~\ref{prop:shallowright} in Appendix).

The inductive cases, $\lseq = \inp A, \lseq'$ and $\lseq = \diabr{\lseq'}$, follow the shallow-to-deep approach~\cite{goetschi_realization_2012} but using intuitionistic reasoning and applying lifting, Lemma~\ref{lem:lifting}, for both proof and satisfier terms (see Lemma~\ref{lem:rightrules} in Appendix).
\qed
\end{proof}

Given an instance of a~$\rle \in \set{\botrule, \orl, \andl, \boxl, \dial, \constrimpl, \cont, \tl, \fourl}$, 
for some $k\le n$,
it will be of the shape:
    $$
    \vliiinf{\rle}{}
    {
        \seq[\fseq]{\lseq_0}
    }
    {
        \seq[\fseq]{\lseq_1}
    }
    {
        \cdots
    }
    {
        \seq[\fseq]{\lseq_k}
    }
    $$
This can be further refined depending of the precise position of the~$\lseq_i$ within context~$\seq{}$.
If the principal formula is among the LHS part of $\seq{}$, the $\rle$ can be rewritten:
    $$
    \vliiinf{\rle}{}
    {
        \seq[\lseq]{\lseq_0}, \oseq
    }
    {
        \seq[\lseq]{\lseq_1}, \oseq
    }
    {
        \cdots
    }
    {
        \seq[\lseq]{\lseq_k}, \oseq
    }
    $$
In this case, we get the following upside-down reading of the previous lemma:

\begin{restatable}{lemma}{lemrealblacklam}
\label{lem:realblacklam}
    %Let $\J \in \set{\JIK, \JIKt, \JIKfour, \JISfour}$.
    %
    For any realisations~$r_i$ on $\seq[\lseq]{\lseq_i}$ for $i\le k$, 
    there exist a realisation~$r$ on~$\seq[\lseq]{\lseq_0}$ 
    and substitutions $\sub_i$ %with $\dom{\sub_j} \sset \negvar{\fseq_j}$ for $j \in \set{1, \dots, i}$ 
    for $i\le k$, 
    such that
    $$
    \proves{\J}
    {
        \real{\seq[\lseq]{\lseq_0}}
        \IMPLIES
        \left(
        \subst[\sub_1]{\real[r_1]{\seq[\lseq]{\lseq_1}}}
        \OR
        \dots
        \OR
        \subst[\sub_k]{\real[r_k]{\seq[\lseq]{\lseq_k}}}
        \right) 
    }
    $$
\end{restatable}

\begin{proof}[Sketch]
    As before, this proceeds by induction on $\seq[\lseq]{}$.
    See Propositions~\ref{prop:LHSbot}, \ref{prop:shalloworl}, \ref{prop:shallowlrules} in Appendix for details.
    Again, in the case of the axiomatic $\botrule$, the empty disjunction reduces to $\BOT$, and in the case of the branching $\orl$, merging (Theorem~\ref{thm:merging}) is required to reconcile the realisations of the branches.
    \qed
\end{proof}

On the other hand, if the principal formula is among the RHS of $\seq{}$, $\rle$ presents itself as
    $$
    \vliiinf{\rle}{}
    {
        \lseq, \seq[\oseq]{\lseq_0}
    }
    {
        \lseq, \seq[\oseq]{\lseq_1}
    }
    {
        \cdots
    }
    {
        \lseq, \seq[\oseq]{\lseq_k}
    }
    $$
which lets us prove a similar lemma to the classical case and Lemma~\ref{lem:realwhite}:

\begin{restatable}{lemma}{lemrealblackpi}
\label{lem:realblackpi}
    %Let $\J \in \set{\JIK, \JIKt, \JIKfour, \JISfour}$.
    %
    For any realisations~$r_i$ on $\seq[\lseq]{\fseq_i}$ for $i\le k$, 
    there exist a realisation~$r$ on~$\seq[\lseq]{\fseq_0}$ 
    and substitutions $\sub_i$ %with $\dom{\sub_j} \sset \negvar{\fseq_j}$ for $j \in \set{1, \dots, i}$ 
    for $i\le k$, 
    such that
    $$
    \proves{\J}
    {
        \left(
        \subst[\sub_1]{\real[r_1]{\seq[\oseq]{\lseq_1}}}
        \AND
        \dots
        \AND
        \subst[\sub_k]{\real[r_k]{\seq[\oseq]{\lseq_k}}}
        \right)
        \IMPLIES
        \real{\seq[\oseq]{\lseq_0}}
    }
    $$
\end{restatable}

\begin{proof}[Sketch]
Similarly as before, the proof proceeds by induction on $\seq[\oseq]{}$.
The base case, when $\seq[\oseq]{} = \out{A},\seq[]{}$ for some $A\in\langann$, depends on the logical content of each rule.
The inductive case, when $\oseq = \br{\seq{}}$, is where the complexity resides.
Indeed, it presents the same dichotomy as described above, where the principal formula can occur in the LHS or the RHS of $\seq{}$.
In the latter, the inductive hypothesis applies readily, while in the former, we need Lemma~\ref{lem:realblacklam} to be able to conclude.
See details in Propositions~\ref{lem:botrule}, \ref{lem:orlrule}, \ref{lem:leftrules}, \ref{lem:implrule} in Appendix.
\qed
\end{proof}

\vspace{1cm}

Finally, putting all the ingredients together,

\begin{proof}[Proof of Theorem~\ref{thm:nestreal}]
    %By Propostion~\ref{prop:nestproofimptrans}, there is a derivation $\pi$ of~$\fseq$ in~$\nIL^{*}$.
	We proceed by induction on the structure of the proof~$\pi$ of $\fseq$ in $\nL$.
	For the base case, $\fseq$ is a conclusion of the~$\id$ or $\botrule$ rule.
	%
	%Apply Lemma~\ref{lem:idrule} or~\ref{lem:botrule} to construct a realisation for~$\fseq$.
	%
	For the inductive case, 
    $\fseq$ is the conclusion of the following rule~$\rle$
	$$
	\vlderivation
	{
		\vliiin{\rle}{}{\fseq}
		{
			\vlhtr{\pi_1}{\fseq_1}
		}
		{\vlhy{\dots}}
		{
			\vlhtr{\pi_i}{\fseq_i}
		}
	}
	$$
	for smaller proofs~$\pi_1, \dots, \pi_k$ and premisses~$\fseq_1, \dots, \fseq_k$ where $i \in \set{1..k}$.
	By the inductive hypothesis, we have realisations $r_i$ on $\fseq_i$ such that
	$$
	\proves{\JIL}{\real[r_i]{\fseq_i}}
	$$
	Applying Lemma~\ref{lem:andrrule}, \ref{lem:orlrule}, \ref{lem:rightrules}, \ref{lem:leftrules} or~\ref{lem:implrule} corresponding to the rule~$\rle$, there exists a realisation~$r$ on~$\fseq$ and substitutions~$\sub_1, \dots, \sub_i$ (which can be $\Idnty$ if not mentioned in the Lemmas) such that
	$$
	\proves{\JIL}
	{
		\subst[\sub_1]{\real[r_1]{\fseq_1}}
		\IMPLIES
		\dots
		\IMPLIES
		\subst[\sub_i]{\real[r_i]{\fseq_i}}
		\IMPLIES
		\real{\fseq}
	}
	$$
	By the Substitution Lemma~\ref{lem:subst}
	$$
		\proves{\JIL}{\subst[\sub_j]{\real[r_j]{\fseq_j}}}
	$$
	for each $j \in \set{1, \dots, i}$.
	Using modus ponens, we achieve
		$$
	\proves{\JIL}
	{
		\real{\fseq}
	}
	$$
    \qed
\end{proof}

We thus obtain the realisation theorem as a corollary:

\begin{proof}[Proof of Theorem~\ref{thm:realisation}]
	By Theorem~\ref{thm:soundcomp}, there is a nested sequent derivation of~$\out{A}$ and we apply Theorem~\ref{thm:nestreal} to construct a realisation.
    \qed
\end{proof}
	
	\section{Conclusion and Future Work}\label{sec:concl}
    We have presented a justification counterpart for intuitionistic modal logic, with a modified merging theorem for realisations and a method for a realisation theorem using intuitionistic nested sequents.

	 \bibliographystyle{splncs04}
	 \bibliography{bibliography}

    \appendix
    \section{Appendix}
    \begin{figure}[t]
	\fbox{
    \begin{minipage}{.95\textwidth}
	\centering
	$\vlinf{\botrule}{}{\seq{\inp{\BOT}}}{} \qquad \vlinf{\id}{}{\seq{\inp{p}, \out{p}}}{}$
	
	$\vlinf{\cont}{}{\seq{\lseq_3}}{\seq{\lseq_1, \lseq_2}}
    \qquad
    \vlinf{\andl}{}{\seq{\inp{A \AND B}}}{\seq{\inp{A}, \inp{B}}} \qquad \vliinf{\andr}{}{\seq{\out{A \AND B}}}{\seq{\out{A}}}{\seq{\out{B}}}$
	
	$\vliinf{\orl}{}{\seq{\inp{A \OR B}}}{\seq{\inp{A}}}{\seq{\inp{B}}} \qquad \vlinf{\orr}{}{\seq{\out{A \OR B}}}{\seq{\out{A}}} \qquad \vlinf{\orr}{}{\seq{\out{A \OR B}}}{\seq{\out{B}}}$
	
	$\vliinf{\impl}{}{\seq{\inp{A \IMPLIES B}}}{\seqprune{\out{A}}}{\seq{\inp{B}}} \qquad \vlinf{\impr}{}{\seq{\out{A \IMPLIES B}}}{\seq{\inp{A}, \out{B}}}$
	
	$\vlinf{\boxlbox}{}{\seq{\inp{\BOX[4j+1] A}, \br[4n]{\oseq}}}{\seq{\br[4n]{\inp{A}, \oseq}}} \quad \vlinf{\boxldia}{}{\seq{\inp{\BOX[4j+1] A}, \diabr[4k+3]{\lseq}}}{\seq{\diabr[4k+3]{\inp{A}, \lseq}}} \qquad \vlinf{\boxr}{}{\seq{\out{\BOX[4n] A}}}{\seq{\br[4n]{\out{A}}}}$
	
	$\vlinf{\dial}{}{\seq{\inp{\DIA[4k+3] A}}}{\seq{\diabr[4k+3]{\inp{A}}}} \qquad \vlinf{\diar}{}{\seq{\out{\DIA[4j+2] A}, \diabr[4k+3]{\lseq}}}{\seq{\br[4n]{\out{A}, \lseq}}}$
    \end{minipage}
    }
	\caption{Annotated system $\nIK$}
	\label{fig:nIKann}
	
	\bigskip
	\fbox{
    \begin{minipage}{.95\textwidth}
	\centering
	$\vlinf{\tl}{}{\seq{\inp{\BOX[4i+1] A}}}{\seq{\inp{A}}} \qquad \vlinf{\tr}{}{\seq{\out{\DIA[4j+2] A}}}{\seq{\out{A}}}$
	
	$\vlinf{\fourlbox}{}{\seq{\inp{\BOX[4j+1] A}, \br[4n]{\oseq}}}{\seq{\br[4n]{\inp{\BOX[4j+1] A}, \oseq}}} \quad \vlinf{\fourldia}{}{\seq{\inp{\BOX[4k+1] A}, \diabr[4i+3]{\lseq}}}{\seq{\diabr[4i+3]{\inp{\BOX[4k+1] A}, \lseq}}} \qquad \vlinf{\fourr}{}{\seq{\out{\DIA[m] A}, \diabr[4k+3]{\lseq}}}{\seq{\br[4n]{\out{\DIA[m] A}, \lseq}}}$
    \end{minipage}
    }
	\caption{Annotated modal rules for $\tax$ and $\fax$ %\Paaras{Can Sonia box this up nicely please?}
    }
	\label{fig:nIS4ann}
\end{figure}

\lemnec*
\begin{proof}
    We proceed by induction on the length of the Hilbert proof of~$\proves{\J}{A}$.

    For the base case, $A$ is either an axiom instance of $\J$ or is an instance of the constant axiom necessitation rule i.e. $A = \just{\pconst[n]}{\just{\dots}{\just{\pconst[1]}{B}}}$ for $B$ an axiom instance of $\J$ and $\pconst[1], \dots, \pconst[n] \in \prfconst$.
    In either case, we can use the constant axiom necessitation rule to deduce
    $$\proves{\J}{\just{\pconst[n+1]}{A}}$$
    for some $\pconst[n+1] \in \prfconst$.
    Set the ground term $t \colonequals \pconst[n+1]$.

    For the inductive case, $A$ is a conclusion of a modus ponens inference of $B$ and $B \IMPLIES A$.
    By the inductive hypothesis, there exists ground terms $t_1$ and $t_2$ such that $\proves{\J}{\just{t_1}{B}}$ and $\proves{\J}{\just{t_2}{(B \IMPLIES A)}}$.
    Using the $\jkax{1}$ axiom
    $$
        \proves{\J}
        {
            \just{t_2}{(B \IMPLIES A)}
            \IMPLIES
            (
                \just{t_1}{B}
                \IMPLIES
                \just{\jappl{t_2}{t_1}}{A}
            )
        }
    $$
    we deduce $\proves{\J}{\just{\jappl{t_2}{t_1}}{A}}$ and set $t = \jappl{t_2}{t_1}$ which is a ground term.
\end{proof}
    
    %\section{Proofs for section~\ref{sec:real}}
    %\subsection{Technical lemmas}

\begin{restatable}[Some facts about realisations and substitutions]{lemma}{facts}
\label{lem:facts}
	\begin{enumerate}
		\item Given a realisation function on $A$ and substitution $\sub$, $\sub \comp r$ is a realisation on~$A \in \langann$ iff $\dom{\sub} \INT \negvar{A} = \EMPTY$, and $\subst{\real{A}} = \real[\sub \comp r]{A}$
		\item If $\dom{r} \INT \dom{r'} \sset \set[2n+1]{n \in \nat}$ then $r \UNI r'$ is a realisation function.
		\item If $A$ and $B$ are uniquely annotated formulas, and $r$ and $r'$ are realisation functions on $A$ and $B$ respectively, $r \UNI r'$ is a realisation function on $A \rhd B$ where $\rhd \in \set{\AND, \OR, \IMPLIES}$. 
		\item If $A$ is an annotated formula and $r$ is a realisation function on $A$, then $r \UNI \set{(4n, t)}$ is a realisation function on $\BOX[4n] A$.
		\item If $A$ is an annotated formula and $r$ is a realisation function on $A$, then $r \UNI \set{(4n+1, \pvar[n])}$ is a realisation function on $\BOX[4n+1] A$.
		\item If $A$ is an annotated formula and $r$ is a realisation function on $A$, then $r \UNI \set{(4k+2, \mu)}$ is a realisation function on $\DIA[4k+2] A$.
		\item If $A$ is an annotated formula and $r$ is a realisation function on $A$ and $\svar[k]$ does not occur in $\real{A}$, then $r \UNI \set{(4k+3, \svar[k])}$ is a realisation function on~$\DIA[4k+2] A$.
	\end{enumerate}
\end{restatable}

%\todo{find space for:}
\begin{definition}
    A shallow instance of a rule is when $\depth{\seq{}} = 0$.
\end{definition}
\begin{example}
    The shallow instance of the $\boxlbox$ rule is 
    $$\vlinf{\boxlbox}{}{{\inp{\BOX A}, \br{\oseq}}}{{\br{\inp{A}, \oseq}}}$$
\end{example}
%\sonia{do we need shallow rules at all?}

%%%%%%%%%%%%%%%%%%%%%%%%%%%%%%%%%%%%%%%%%%%%%%%%%%%%%%%%%%%%%%%%%%%%%%%%%%%%%%%%%%%%%%
%%%%%%%%%%%%%%%%%%%%%%%%%%%%%%%%%%%%%%%%%%%%%%%%%%%%%%%%%%%%%%%%%%%%%%%%%%%%%%%%%%%%%%
%%%%%%%%%%%%%%%%%%%%%%%%%%%%%%%%%%%%%%%%%%%%%%%%%%%%%%%%%%%%%%%%%%%%%%%%%%%%%%%%%%%%%%
\subsection{Right rules}
%%%%%%%%%%%%%%%%%%%%%%%%%%%%%%%%%%%%%%%%%%%%%%%%%%%%%%%%%%%%%%%%%%%%%%%%%%%%%%%%%%%%%%
%%% ID-RULE
\begin{restatable}[id-rule]{lemma}{lemidrule}
\label{lem:idrule}
	Let~$\J \in \set{\JIK, \JIKt, \JIKfour, \JISfour}$.
	%
	%Given $\vlinf{\id}{}{\seq{\inp{p}, \out{p}}}{}$.
	%
        Let $\seq{\inp{p}, \out{p}}$ be the conclusion of the $\id$ rule.
	Then there exists a realisation function~$r$ on~$\seq{\inp{p}, \out{p}}$ such that
	$$\proves{\J}{\real{\seq{\inp{p}, \out{p}}}}$$
\end{restatable}

\begin{proof}
	We proceed by induction on~$\depth{\seq{\ }}$.
	
	For the base case~$\depth{\seq{\ }}=0$, we have $\seq{\inp{p}, \out{p} } = \lseq, \inp{p}, \out{p}$ for some annotated LHS sequent~$\lseq$.
	Define
	\begin{multline*}
		r \colonequals \set[(4m, {\respvar[m]})]{\text{$4m$ is an annotation in~$\lseq$}} 
		\\ \UNI \set[(4m+1, {\pvar[m]})]{\text{$4m+1$ is an annotation in~$\lseq$}} 
		\\ \UNI \set[(4m+2, {\ressvar[m]})]{\text{$4m+2$ is an annotation in~$\lseq$}} 
		\\ \UNI \set[(4m+3, {\svar[m]})]{\text{$4m+3$ is an annotation in~$\lseq$}}
	\end{multline*}
	so negative annotations are realised with variables as standard and the positive annotations are realised with the reserved variables.
	It follows as a theorem of~$\IPL$ that
	$$
	\proves{\J}
	{
		\real{\seq{\inp{p}, \out{p}}}
		=
		(\real{\lseq} \AND p)
		\IMPLIES
		p
	}
	$$
	
	For the inductive case~$\seq{\inp{p}, \out{p} } = \lseq, \br[4n]{\seq[\fseq']{\inp{p}, \out{p} }}$ for some annotated LHS sequent~$\lseq$ and sequent~$\seq[\fseq']{\inp{p}, \out{p} }$.
	By the inductive hypothesis, there exists a realisation~$r'$ on~$\seq[\fseq']{\inp{p}, \out{p} }$ such that
	$$
	\proves{\J}
	{
		\real[r']{\seq[\fseq']{\inp{p}, \out{p} }}
	}
	$$
	
	By Lemma~\ref{lem:nec}, there exists a ground proof term~$t$ such that
	$$
	\proves{\J}
	{
		\just{t}{\real[r']{\seq[\fseq']{\inp{p}, \out{p} }}}
	}
	$$
	Similarly, define
	\begin{multline*}
		r \colonequals r' \UNI \set{(4n, t)} \UNI \set[(4m, {\respvar[m]})]{\text{$4m$ is an annotation in~$\lseq$}} 
		\\ \UNI \set[(4m+1, {\pvar[m]})]{\text{$4m+1$ is an annotation in~$\lseq$}} 
		\\ \UNI \set[(4m+2, {\ressvar[m]})]{\text{$4m+2$ is an annotation in~$\lseq$}} 
		\\ \UNI \set[(4m+3, {\svar[m]})]{\text{$4m+3$ is an annotation in~$\lseq$}}
	\end{multline*}
	and by propositional reasoning we have
	$$
	\proves{\J}
	{
		\real{\seq{\inp{p}, \out{p} }} =
	(	\real{\lseq} \IMPLIES \just{t}{\real[r']{\seq[\fseq']{\inp{p}, \out{p} }}})
	}
	$$
    \qed
\end{proof}

%%%%%%%%%%%%%%%%%%%%%%%%%%%%%%%%%%%%%%%%%%%%%%%%%%%%%%%%%%%%%%%%%%%%%%%%%%%%%%%%%%%%%%
%%% AND-RIGHT
\begin{restatable}[$\AND$-right rule]{lemma}{lemandrule}
\label{lem:andrrule}
	Let~$\J \in \set{\JIK, \JIKt, \JIKfour, \JISfour}$.
	Let~$A, B \in \langann$.
	Let~$\seq{\out{A}}$ and~$\seq{\out{B}}$ be the premisses and~$\seq{\out{A \AND B}}$ be the conclusion sequents of an annotated instance of the~$\andr$ rule.
	Then, given realisation functions~$r_1$ on~$\seq{\out{A}}$ and~$r_2$ on~$\seq{\out{B}}$, there exists a realisation function~$r$ and a substitution~$\sub$ 
    %with~$\dom{\sub} \sset \negvar{\seq{\out{A \AND B}}}$ 
    %\sonia{Should this inclusion be the other way around?}
    %
    such that
	$$\proves{\J}{\subst{\real[r_1]{\seq{\out{A}}}} \IMPLIES \subst{\real[r_2]{\seq{\out{B}}}} \IMPLIES \real{\seq{\out{A \AND B}}}}$$
\end{restatable}
\begin{proof}
	We proceed by induction on~$\depth{\seq{\ }}$.
	
	For the base case~$\depth{\seq{\ }}=0$, we have $\seq{\ } = \lseq, \{\ \}$ for some annotated LHS sequent~$\lseq$.
	By Corollary~\ref{cor:nestedmerging}, there exists a realisation $r'$ on~$\lseq$ and substitution~$\sub$ with $\dom{\sub} \sset \negvar{\lseq}$ such that
	$$\proves{\J}{\real[r']{\lseq} \IMPLIES \subst{\real[r_i]{\lseq}}}$$
	for~$i \in \set{1, 2}$.
	As~$\negvar{\lseq} \INT \negvar{A} = \EMPTY$ and $\negvar{\lseq} \INT \negvar{B} = \EMPTY$, by Lemma~\ref{lem:facts}, $\sub \comp r_1$ and $\sub \comp r_2$ are realisation functions on A and B respectively, and $r \colonequals r' \UNI \restr[(\sub \circ r_1)]{A} \UNI \restr[(\sub \circ r_2)]{B}$ is a realisation function on~$\lseq, \out{A \AND B}$ and we have
	$$\proves{\J}{\subst{\real[r_1]{A}} \IMPLIES \subst{\real[r_2]{B}} \IMPLIES \real{A} \AND \real{B}}$$
	and it follows from propositional reasoning that
	$$\proves{\J}{\subst{\real[r_1]{(\lseq, \out{A})}} \IMPLIES \subst{\real[r_2]{(\lseq, \out{B})}} \IMPLIES \real{(\lseq, \out{A \AND B})}}$$
	
	For the inductive case~$\seq{\ } = \lseq, \br[4n]{\seq[\fseq']{\ }}$ for some annotated LHS sequent~$\lseq$ and context~$\seq[\fseq']{\ }$.
	Using the inductive hypothesis, there exists a realisation $r'$ on $\seq[\fseq']{\out{A \AND B}}$ and substitution~$\sub'$ with~$\dom{\sub} \sset \negvar{\seq[\fseq']{\out{A \AND B}}}$ such that
	$$\proves{\J}{\subst[\sub']{\real[r_1]{\seq[\fseq']{\out{A}}}} \IMPLIES \subst[\sub']{\real[r_2]{\seq[\fseq']{\out{B}}}} \IMPLIES \real[r']{\seq[\fseq']{\out{A \AND B}}}}$$
	By Lemma~\ref{lem:facts}, $\restr[(\sub' \comp r_i)]{\lseq}$ is a realisation on~$\lseq$ for each~$i \in \set{1,2}$.
	By Corollary~\ref{cor:nestedmerging}, there exists a realisation~$r''$ on~$\lseq$ and substitution~$\sub''$ with~$\dom{\sub''} \sset \negvar{\lseq}$ such that
	$$\proves{\J}{\real[r'']{\lseq} \IMPLIES \subst[\sub'']{\real[\sub' \comp r_i]{\lseq}}}$$
	for~$i \in \set{1, 2}$.
	By the Substitution Lemma~\ref{lem:subst}
	$$
		\proves{\J}{\subst[(\sub'' \comp \sub')]{\real[r_1]{\seq[\fseq']{\out{A}}}} \IMPLIES \subst[(\sub'' \comp \sub')]{\real[r_2]{\seq[\fseq']{\out{B}}}} \IMPLIES \subst[\sub'']{\real[r']{\seq[\fseq']{\out{A \AND B}}}}}
	$$
	Set~$\sub = \sub'' \comp \sub$.
	By the Lifting Lemma~\ref{lem:lifting}, there exists a proof term~$t \equiv t((\sub \circ r_1)(4n), (\sub \circ r_2)(4n))$ such that
	$$\proves{\J}{\just{(\sub \circ r_1)(4n)}{\subst{\real[r_1]{\seq[\fseq']{\out{A}}}}} \IMPLIES \just{(\sub \circ r_2)(4n)}{\subst{\real[r_1]{\seq[\fseq']{\out{B}}}}} \IMPLIES 
		\just{t}
		{\real[\sub'' \comp r']{\seq[\fseq']{\out{A \AND B}}}}
	}
	$$
	Set~$r \colonequals (\sub'' \comp r') \UNI r'' \UNI \set{(4n, t)}$ which is a realisation function on $(\lseq, \br[4n]{\seq[\fseq']{\out{A \AND B}}})$ by Lemma~\ref{lem:facts} and by propositional reasoning
	$$\proves{\J}{\subst{\real[r_1]{(\lseq, \br[4n]{\seq[\fseq']{\out{A}}})}} \IMPLIES \subst{\real[r_2]{(\lseq, \br[4n]{\seq[\fseq']{\out{B}}})}} \IMPLIES \real{(\lseq, \br[4n]{\seq[\fseq']{\out{A \AND B}}})}}$$
    \qed
\end{proof}

%%%%%%%%%%%%%%%%%%%%%%%%%%%%%%%%%%%%%%%%%%%%%%%%%%%%%%%%%%%%%%%%%%%%%%%%%%%%%%%%%%%%%%
%%%%%%%%%%%%%%%%%%%%%%%%%%%%%%%%%%%%%%%%%%%%%%%%%%%%%%%%%%%%%%%%%%%%%%%%%%%%%%%%%%%%%%
%%% RIGHT-RULES
\begin{restatable}[Other right rules -- base case]{proposition}{propshallowright}\label{prop:shallowright}
    %\sonia{is this proposition supposed to be stated as the shallow version? it is not clear from the statement but from the proof it seems to be the case}
	Let $\J \in \set{\JIK, \JIKt, \JIKfour, \JISfour}$.
	Let~$\fseq'$ be the premiss and $\fseq$ be the conclusion sequents of an annotated shallow instance of a rule~$\rle \in \set{\orr, \impr, \boxr, \diar, \boxlbox, \tr, \fourr, \fourlbox, \updtrle}$. 
	Then, given a realisation function~$r'$ on~$\fseq'$, there exists a realisation function~$r$ on~$\fseq$ such that:
    $\proves{\J}{\real[r']{\fseq'} \IMPLIES \real{\fseq}}$
    with $\J$ appropriately containing the corresponding axioms to the modal rule considered.
    %
%	\begin{itemize}
%		\item $\proves{\J}{\real[r']{\fseq'} \IMPLIES \real{\fseq}}$ if~$\rle \in \set{\orr, \impr, \boxr, \diar, \boxlbox, \updtrle}$;
%		\item $\proves{\ILP}{\real[r']{\fseq'} \IMPLIES \real{\fseq}}$ if~$\rle \in \set{\tr, \fourr, \fourlbox}$.
%	\end{itemize}
\end{restatable}
%\Paaras{Cases 1, 2, 3 are "trivial"}
\begin{proof}\setcounter{case}{0}
	\begin{case}
		When~$\rle = \orr$.
		We will only consider one case as the other case is similar.
		We have~$\fseq' = \lseq, \out{A}$ and~$\fseq = \lseq, \out{A \OR B}$ for some annotated LHS sequent~$\lseq$.
		Extend~$r'$ onto~$B$ similar to Lemmas~\ref{lem:idrule} and~\ref{lem:botrule}.
		By propositional reasoning, we have
		$$\proves{\J}{(\real[r']{\lseq} \IMPLIES \real[r']{A}) \IMPLIES (\real[r']{\lseq} \IMPLIES \real[r']{(A \OR B)})}$$
		and so we set~$r \colonequals r'$.
	\end{case}
	\begin{case}
		When~$\rle = \impr$.
		We have~$\fseq' = \lseq, \inp{A}, \out{B}$ and~$\fseq = \lseq, \out{A \IMPLIES B}$ for some annotated LHS sequent~$\lseq$.
		By propositional reasoning, we have
		$$\proves{\J}{((\real[r']{\lseq} \AND \real[r']{A}) \IMPLIES \real[r']{B}) \IMPLIES (\real[r']{\lseq} \IMPLIES (\real[r']{A} \IMPLIES \real[r']{B}))}$$
		and so we set~$r \colonequals r'$.
	\end{case}
	\begin{case}
		When~$\rle = \boxr$.
		We have~$\fseq' = \lseq, \br[4n]{\out{A}}$ and~$\fseq = \lseq, \out{\BOX[4n] A}$.
		We note~$\real[r']{\fseq'} = \real[r']{\fseq}$ so we set~$r \colonequals r'$.
	\end{case}
	\begin{case}
		When~$\rle = \diar$.
		We have $\fseq' = \lseq', \br[4n]{\out{A}, \lseq}$ and~$\fseq = \lseq', \out{\DIA[4j+2] A}, \diabr[4k+3]{A}$ for some annotated LHS sequents~$\lseq$ and~$\lseq'$.
		We have
		$$\real[r']{\br[4n]{\out{A}, \lseq}} = \just{r'(4n)}{(\real[r']{\lseq} \IMPLIES \real[r']{A})}$$
		Using the~$\jkax{2}$ axiom we have
		$$\proves{\J}{\just{r'(4n)}{(\real[r']{\lseq} \IMPLIES \real[r']{A})} \IMPLIES (\sat{\svar[k]}{\real[r']{\lseq}} \IMPLIES \sat{(\sappl{r'(4n)}{\svar[k]})}{\real[r']{A}})}$$
		By propositional reasoning, we have
		$$\proves{\J}{(\real[r']{\lseq'} \IMPLIES \just{r'(4n)}{(\real[r']{\lseq} \IMPLIES \real[r']{A})}) \IMPLIES ((\real[r']{\lseq'} \AND \sat{\svar[k]}{\real[r']{\lseq}}) \IMPLIES \sat{(\sappl{r'(4n)}{\svar[k]})}{\real[r']{A}})}$$
		and so we set~$r \colonequals \restr[r']{\lseq'} \UNI \set{(4j+2, \sappl{r'(4n)}{\svar[m]}), (4k+3, \svar[k])}$ which is a realisation function on~$\fseq$ by Lemma~\ref{lem:facts}.
	\end{case}
        \begin{case}
            When~$\rle = \boxlbox$. $\fseq' = \lseq, \br[4n]{\inp{A}, \oseq}$ and $\fseq = \lseq, \inp{\BOX[4j+1] A}, \br[4n]{\oseq}$.
            We have
            $$\real[r']{\fseq'} = \real[r']{\lseq} \IMPLIES \just{r'(4n)}{(\real[r']{A} \IMPLIES \real[r']{\oseq})}$$
            Using the $\jkax{1}$ axiom,
            $$
            \proves{\J}{\just{r'(4n)}{(\real[r']{A} \IMPLIES \real[r']{\oseq})}
            \IMPLIES
            (\just{\pvar[j]}{\real[r']{A}}
            \IMPLIES
            \just{(\jappl{r'(4n)}{\pvar[j]})}{\real[r']{\oseq}})}
            $$
            Propositional reasoning gives
            $$
            \proves{\J}{
            (\real[r']{\lseq} \IMPLIES
            \just{r'(4n)}
            {
            (\real[r']{A} \IMPLIES \real[r']{\oseq})})
            \IMPLIES
            ((\real[r']{\lseq} \AND
            \just{\pvar[j]}{\real[r']{A}})
            \IMPLIES
            \just{(\jappl{r'(4n)}{\pvar[j]})}{\real[r']{\oseq}})
            }
            $$
            Set $r \colonequals \restr[r']{\lseq} \UNI \restr[r']{A} \UNI \restr[r']{\oseq} \UNI \set{(4j+1, \pvar[j]), (4n, \jappl{r'(4n)}{\pvar[j]})}$ which is a realisation function on $A$ by Lemma~\ref{lem:facts}.
        \end{case}
	\begin{case}
		When~$\rle = \tr$.
		We have~$\fseq' = \lseq, \out{A}$ and~$\fseq = \lseq, \out{\DIA[4j+2] A}$ for some annotated LHS sequent~$\lseq$.
		Note~$\real[r']{(\lseq, \out{A})} = \real[r']{\lseq} \IMPLIES \real[r']{A}$.
		Using the~$\jtaxd$ axiom and propositional reasoning, we have
		$$\proves{\ILP}{(\real[r']{\lseq} \IMPLIES \real[r']{A}) \IMPLIES (\real[r']{\lseq} \IMPLIES \sat{\ressvar[j]}{\real[r']{A}})}$$
		where we use a reserved variable.
		Set $r \colonequals r \UNI \set{(4j+2, \ressvar[j])}$ which is a realisation function on~$\fseq$ by Lemma~\ref{lem:facts}.
	\end{case}
	\begin{case}
		When~$\rle = \fourr$.
		We have $\fseq' = \lseq', \br[4n]{\out{\DIA[m] A}, \lseq}$ and~$\fseq = \lseq', \out{\DIA[m] A}, \diabr[4k+3]{A}$ for some annotated LHS sequents~$\lseq$ and~$\lseq'$.
		We have
		$$\real[r']{\br[4n]{\out{\DIA[m] A}, \lseq}} = \just{r'(4n)}{(\real[r']{\lseq} \IMPLIES \sat{r'(m)}{\real[r']{A}})}$$
		Using the~$\jkax{2}$ axiom we have
		$$\proves{\ILP}{\just{r'(4n)}{(\real[r']{\lseq} \IMPLIES \sat{r'(m)}{\real[r']{A}})} \IMPLIES (\sat{\svar[k]}{\real[r']{\lseq}} \IMPLIES \sat{(\sappl{r'(4n)}{\svar[k]})}{\sat{r'(m)}{\real[r']{A}}})}$$
		Using the~$\jfaxd$ axiom and propositional reasoning
		$$\proves{\ILP}{\just{r'(4n)}{(\real[r']{\lseq} \IMPLIES \sat{r'(m)}{\real[r']{A}})} \IMPLIES (\sat{\svar[k]}{\real[r']{\lseq}} \IMPLIES \sat{r'(m)}{\real[r']{A}})}$$
		By further propositional reasoning, we have
		$$\proves{\ILP}{(\real[r']{\lseq'} \IMPLIES \just{r'(4n)}{(\real[r']{\lseq} \IMPLIES \real[r']{A})}) \IMPLIES ((\real[r']{\lseq'} \AND \sat{\svar[k]}{\real[r']{\lseq}}) \IMPLIES \sat{r'(m)}{\real[r']{A}})}$$
		and so we set~$r \colonequals \restr[r']{\lseq'} \UNI \set{(4k+3, \svar[m])}$ which is a realisation function on~$\fseq$ by Lemma~\ref{lem:facts}.
	\end{case}
        \begin{case}
            When~$\rle = \fourlbox$. $\fseq' = \lseq, \br[4n]{\inp{\BOX[4j+1] A}, \oseq}$ and $\fseq = \lseq, \inp{\BOX[4j+1] A}, \br[4n]{\oseq}$.
            We have
            $$\real[r']{\fseq'} = \real[r']{\lseq} \IMPLIES \just{r'(4n)}{(\just{\pvar[j]}{\real[r']{A}} \IMPLIES \real[r']{\oseq})}$$
            Using the $\jkax{1}$ axiom,
            $$
            \proves{\J}{\just{r'(4n)}{(\just{\pvar[j]}{\real[r']{A}} \IMPLIES \real[r']{\oseq})}
            \IMPLIES
            (\just{\jbang{\pvar[j]}}{\just{\pvar[j]}{\real[r']{A}}}
            \IMPLIES
            \just{(\jappl{r'(4n)}{\jbang{\pvar[j]}})}{\real[r']{\oseq}})}
            $$
            The $\jfaxb$ axiom gives
            $$
            \proves{\J}{\just{\pvar[j]}{A} \IMPLIES \just{\jbang{\pvar[j]}}{\just{\pvar[j]}{A}}}
            $$
            Propositional reasoning gives
            $$
            \proves{\J}{
            (\real[r']{\lseq} \IMPLIES
            \just{r'(4n)}
            {
            (\real[r']{A} \IMPLIES \real[r']{\oseq})})
            \IMPLIES
            ((\real[r']{\lseq} \AND
            \just{\pvar[j]}{\real[r']{A}})
            \IMPLIES
            \just{(\jappl{r'(4n)}{\jbang{\pvar[j]}})}{\real[r']{\oseq}})
            }
            $$
            Set $r \colonequals \restr[r']{\lseq} \UNI \restr[r']{A} \UNI \restr[r']{\oseq} \UNI \set{(4j+1, \pvar[j]), (4n, \jappl{r'(4n)}{\jbang{\pvar[j]}})}$ which is a realisation function on $A$ by Lemma~\ref{lem:facts}.
        \end{case}
        \begin{case}
            When $\rle = \updtrle$.
            $\fseq' = \diabr[4k+3]{\lseq_1}, \br[4n]{\lseq_2, \oseq}$ and $\fseq = \br[4n]{\lseq_1, \lseq_2, \oseq}$. 
            Note
	$$
	\real[r']{\fseq'}
	=
	\sat{\svar[k]}{\real[r']{\lseq_1}} \IMPLIES \just{r'(4n)}{(\real[r']{\lseq_2} \IMPLIES \real[r']{\oseq})}
	$$
	Using the $\jkax{4}$ axiom,
	$$
	\proves{\J}
	{
		(\sat{\svar[k]}{\real[r']{\lseq_1}} \IMPLIES \just{r'(4n)}{(\real[r']{\lseq_2} \IMPLIES \real[r']{\oseq})})
		\IMPLIES
		\just{\jupdt{\svar[k]}{r'(4n)}}
		{
			(\real[r']{\lseq_1} \IMPLIES (\real[r']{\lseq_2} \IMPLIES \real[r']{\oseq}))
		}
	}
	$$
	Propositional reasoning gives
	$$
	\proves{\J}
	{
		(\real[r']{\lseq_1} \IMPLIES (\real[r']{\lseq_2} \IMPLIES \real[r']{\oseq}))
		\IMPLIES
		((\real[r']{\lseq_1} \AND \real[r']{\lseq_2}) \IMPLIES \real[r']{\oseq})
	}
	$$
	Using the Lifting Lemma~\ref{lem:lifting}, there exists a proof term $t \equiv t(\jupdt{\svar[k]}{r(4n)})$ such that
	$$
	\proves{\J}
	{
		\just{\jupdt{\svar[k]}{r'(4n)}}
		{
			(\real[r']{\lseq_1} \IMPLIES (\real[r']{\lseq_2} \IMPLIES \real[r']{\oseq}))
		}
		\IMPLIES
		\just{t}
		{
			((\real[r']{\lseq_1} \AND \real[r']{\lseq_2}) \IMPLIES \real[r']{\oseq})
		}
	}
	$$
	Modus ponens achieves
	$$
	\proves{\J}
	{
		(\sat{\svar[k]}{\real[r']{\lseq_1}} \IMPLIES \just{r'(4n)}{(\real[r']{\lseq_2} \IMPLIES \real[r']{\oseq})})
		\IMPLIES
		\just{t}
		{
			((\real[r']{\lseq_1} \AND \real[r']{\lseq_2}) \IMPLIES \real[r']{\oseq})
		}
	}
	$$
	Set $r \colonequals \restr[r']{\lseq_1} \UNI \restr[r']{\lseq_2} \UNI \restr[r']{\oseq} \UNI \set{(4n,t)}$ which is a realisation function on~$\fseq$ by Lemma~\ref{lem:facts}.
        \end{case}
    \qed
\end{proof}
%%%%%%%%%%%%%%%%%%%%%%%%%%%%%%%%%%%%%%%%%%%%%%%%%%%%%%%%%%%%%%%%%%%%%%%%%%%%%%%%%%%%%%

\begin{restatable}[Other right rules -- inductive case]{lemma}{lemrightrules}
\label{lem:rightrules}
    %\sonia{as stated, this is the same as proposition above}
	Let $\J \in \set{\JIK, \JIKt, \JIKfour, \JISfour}$.
	Let~$\seq{\ }$ be an annotated context where~$\seq{\Omega'}$ is the premiss and~$\seq{\Omega}$ is the conclusion sequent of an annotated instance of a rule~$\rle \in \set{\orr, \impr, \boxr, \diar, \boxlbox, \tr, \fourr, \fourlbox, \updtrle}$ where~$\Omega'$ and~$\Omega$ are full sequents. 
	Then, given realisation function~$r'$ on~$\fseq'$, there exists a realisation function~$r$ on~$\fseq$ such that:
    $\proves{\J}{\real[r']{\seq{\Omega'}} \IMPLIES \real{\seq{\Omega}}}$
    with $\J$ appropriately containing the corresponding axioms to the modal rule considered.
    
%	\begin{itemize}
%		\item $\proves{\J}{\real[r']{\seq{\Omega'}} \IMPLIES \real{\seq{\Omega}}}$ if~$\rle \in \set{\orr, \impr, \boxr, \diar, \boxlbox, \updtrle}$;
%		\item $\proves{\ILP}{\real[r']{\seq{\Omega'}} \IMPLIES \real{\seq{\Omega}}}$ if~$\rle \in \set{\tr, \fourr, \fourlbox}$.
%	\end{itemize}
\end{restatable}
\begin{proof}
	We proceed by induction on~$\depth{\seq{\ }}$.
	Note if~$\rle \in \set{\tr, \fourr}$, then $\J = \ILP$.
	
	For the base case~$\depth{\seq{\ }}=0$, we have $\seq{\ } = \lseq', \{\ \}$ for some annotated LHS sequent~$\lseq'$ and we can construct a realisation~$r$ using Proposition~\ref{prop:shallowright}.
	
	For the inductive case~$\seq{\ } = \lseq, \br[4n]{\seq[\fseq']{\ }}$ for some annotated LHS sequent~$\lseq$ and context~$\seq[\fseq']{\ }$.
	By the inductive hypothesis, there exists a realisation~$r''$ such that
	$$\proves{\J}{\real[r']{\seq[\fseq']{\Omega'}} \IMPLIES \real[r'']{\seq[\fseq']{\Omega}}}$$
	By the Lifting Lemma~\ref{lem:lifting}, there exists a proof term~$t(r'(4n))$ such that
	$$\proves{\J}{\just{r'(4n)}{\real[r']{\seq[\fseq']{\Omega'}}} \IMPLIES \just{t(r'(4n))}{\real[r'']{\seq[\fseq']{\Omega}}}}$$
	By propositional reasoning, we have
	$$\proves{\J}{(\real[r']{\lseq} \IMPLIES \just{r'(4n)}{\real[r']{\seq[\fseq']{\Omega'}}}) \IMPLIES (\real[r']{\lseq} \IMPLIES \just{t(r'(4n))}{\real[r'']{\seq[\fseq']{\Omega}}})}$$
	As~$\dom{r''} \INT \dom{\restr[r']{\lseq}} \sset \set[2n+1]{n \in \nat}$, set~$r \colonequals r'' \UNI \restr[r']{\lseq} \UNI \set{(4n, t(r'(4n)))}$ which is a realisation function on~$\seq{\Omega}$ by Lemma~\ref{lem:facts}.
    \qed
\end{proof}

%%%%%%%%%%%%%%%%%%%%%%%%%%%%%%%%%%%%%%%%%%%%%%%%%%%%%%%%%%%%%%%%%%%%%%%%%%%%%%%%%%%%%%
%%%%%%%%%%%%%%%%%%%%%%%%%%%%%%%%%%%%%%%%%%%%%%%%%%%%%%%%%%%%%%%%%%%%%%%%%%%%%%%%%%%%%%
%%%%%%%%%%%%%%%%%%%%%%%%%%%%%%%%%%%%%%%%%%%%%%%%%%%%%%%%%%%%%%%%%%%%%%%%%%%%%%%%%%%%%%

\subsection{Left rules}
%%%%%%%%%%%%%%%%%%%%%%%%%%%%%%%%%%%%%%%%%%%%%%%%%%%%%%%%%%%%%%%%%%%%%%%%%%%%%%%%%%%%%%
%%% BOT-LEFT
\begin{restatable}[$\BOT$-left rule -- bottom-up]{proposition}{propLHSbot}
\label{prop:LHSbot}
    %\sonia{what is the role of this proposition?}
	Let~$\J \in \set{\JIK, \JIKt, \JIKfour, \JISfour}$.
	Let~$\seq[\lseq]{\inp{\BOT}}$ be an LHS sequent.
	Let~$r$ be a realisation function on~$\seq[\lseq]{\inp{\BOT}}$.
	Then
	$$
	\proves{\J}
	{
		\real{\seq[\lseq]{\inp{\BOT}}} \IMPLIES \BOT
	}
	$$
\end{restatable}
\begin{proof}
	We proceed by induction on~$\depth{\seq[\lseq]{\ }}$.
	
	For the base case~$\depth{\seq[\lseq]{\ }} =0$, we have $\seq[\lseq]{\inp{\BOT}} = \lseq', \inp{\BOT}$ for some annotated LHS sequent~$\lseq$.
	It follows as a theorem of~$\IPL$ that
	$$
	\proves{\J}
	{
		\real{\seq[\lseq]{\inp{\BOT}}}
		=
		(\real{\lseq'} \AND \BOT)
		\IMPLIES
		\BOT
	}
	$$
	
	For the inductive case~$\seq[\lseq]{\inp{\BOT} } = \lseq_0, \diabr[4k+3]{\lseq_1, \seq[\lseq_2]{\inp{\BOT} }}$.
	By the inductive hypothesis and propositional reasoning
	$$
	\proves{\J}
	{	
		(\real{\lseq_1}
		\AND
		\real{\seq[\lseq_2]{\BOT }}) 
		\IMPLIES \BOT
	}
	$$
	By the Lifting Lemma~\ref{lem:lifting}, there exists a satisfier term~$\mu$ such that
	$$
	\proves{\J}
	{	
		\sat{\svar[k]}{(\real{\lseq_1}
		\AND
		\real{\seq[\lseq_2]{\BOT }})} 
		\IMPLIES \sat{\mu}{\BOT}
	}
	$$
	Using the~$\jkax{3}$ axiom and transitivity, we have
		$$
	\proves{\J}
	{	
		\sat{\svar[k]}{(\real{\lseq_1}
			\AND
			\real{\seq[\lseq_2]{\BOT }})} 
		\IMPLIES \BOT
	}
	$$
	By further propositional reasoning we achieve
	$$
	\proves{\J}
	{
		\real{\seq[\lseq]{\inp{\BOT} }} =
		(
			\real{\lseq_0}
			\AND
			\sat{\svar[k]}{(\real{\lseq_1}
				\AND
				\real{\seq[\lseq_2]{\BOT }})} 
		)
		\IMPLIES
		\BOT
	}$$
    \qed
\end{proof}

%%%%%%%%%%%%%%%%%%%%%%%%%%%%%%%%%%%%%%%%%%%%%%%%%%%%%%%%%%%%%%%%%%%%%%%%%%%%%%%%%%%%%%
\begin{restatable}[$\BOT$-left rule -- top-down]{lemma}{lembotrule}
\label{lem:botrule}
	Let~$\J \in \set{\JIK, \JIKt, \JIKfour, \JISfour}$.
        Let $\seq{\inp{\BOT}}$ be the conclusion of the $\botrule$ rule.
	Then there exists a realisation function~$r$ on~$\seq{\inp{\BOT}}$ such that
	$$\proves{\J}{\real{\seq{\inp{\BOT}}}}$$
\end{restatable}

\begin{proof}
First note that the conclusion sequent is of the form $\seq{\seq[\lseq]{\inp{\BOT}}, \oseq}$ for some annotated LHS sequent~$\seq[\lseq]{\inp{\BOT}}$ and annotated RHS sequent~$\oseq$.

We proceed by induction on~$\depth{\seq{\ }}$.

For the base case~$\depth{\seq{\ }}=0$, we have $\seq{\seq[\lseq]{\inp{\BOT}}, \oseq} = \seq[\lseq]{\inp{\BOT}}, \oseq$. 
%\sonia{should it not be: $\lseq, \inp\BOT, \oseq$? would it not simplify the proof?}
%
Define
\begin{multline*}
	r \colonequals \set[(4m, {\respvar[m]})]{\text{$4m$ is an annotation in~$\seq[\lseq]{\inp{\BOT}}$ or~$\oseq$}} 
	\\ \UNI \set[(4m+1, {\pvar[m]})]{\text{$4m+1$ is an annotation in~$\seq[\lseq]{\inp{\BOT}}$ or~$\oseq$}} 
	\\ \UNI \set[(4m+2, {\ressvar[m]})]{\text{$4m+2$ is an annotation in~$\seq[\lseq]{\inp{\BOT}}$ or~$\oseq$}} 
	\\ \UNI \set[(4m+3, {\svar[m]})]{\text{$4m+3$ is an annotation in~$\seq[\lseq]{\inp{\BOT}}$ or~$\oseq$}}
\end{multline*}
so negative annotations are realised with variables as standard and the positive annotations are realised with the reserved variables.
By Proposition~\ref{prop:LHSbot}, $\proves{\J}{\real{\seq[\lseq]{\inp{\BOT}}} \IMPLIES \BOT}$ and by intuitionistic reasoning~$\proves{\J}{\BOT \IMPLIES \real{\oseq}}$.
Hence
$$
\proves{\J}
{
	\real{\seq{\inp{\BOT}}}
	=
	(\real{\seq[\lseq]{\inp{\BOT}}} \IMPLIES \real{\oseq})
}
$$

For the inductive case~$\seq{\seq[\lseq]{\inp{\BOT}}, \oseq} = \lseq, \br[4n]{\seq[\fseq']{\seq[\lseq]{\inp{\BOT}}, \oseq}}$ for some annotated LHS sequent~$\lseq$ and sequent~$\seq[\fseq']{\seq[\lseq]{\inp{\BOT}}, \oseq}$.
By the inductive hypothesis, there exists a realisation~$r'$ on~$\seq[\fseq']{\seq[\lseq]{\inp{\BOT}}, \oseq}$ such that
$$
\proves{\J}
{
	\real[r']{\seq[\fseq']{\seq[\lseq]{\inp{\BOT}}, \oseq}}
}
$$
By Lemma~\ref{lem:nec}, there exists a ground proof term~$t$ such that
$$
\proves{\J}
{
	\just{t}{\real[r']{\seq[\fseq']{\seq[\lseq]{\inp{\BOT}}, \oseq}}}
}
$$
Similarly, define
\begin{multline*}
	r \colonequals r' \UNI \set{(4n, t)} \UNI \set[(4m, {\respvar[m]})]{\text{$4m$ is an annotation in~$\lseq$}} 
	\\ \UNI \set[(4m+1, {\pvar[m]})]{\text{$4m+1$ is an annotation in~$\lseq$}} 
	\\ \UNI \set[(4m+2, {\ressvar[m]})]{\text{$4m+2$ is an annotation in~$\lseq$}} 
	\\ \UNI \set[(4m+3, {\svar[m]})]{\text{$4m+3$ is an annotation in~$\lseq$}}
\end{multline*}
and by propositional reasoning we have
$$
\proves{\J}
{
	\real{\seq{\inp{\seq[\lseq]{\inp{\BOT}}, \oseq}}} =
	(	\real{\lseq} \IMPLIES \just{t}{\real[r']{\seq[\fseq']{\seq[\lseq]{\inp{\BOT}}, \oseq}}})
}
$$
\qed
\end{proof}
%\Paaras{The inductive case of Lemma \ref{lem:botrule} is the same as the $\id$ rule}

%%%%%%%%%%%%%%%%%%%%%%%%%%%%%%%%%%%%%%%%%%%%%%%%%%%%%%%%%%%%%%%%%%%%%%%%%%%%%%%%%%%%%%
%%%%%%%%%%%%%%%%%%%%%%%%%%%%%%%%%%%%%%%%%%%%%%%%%%%%%%%%%%%%%%%%%%%%%%%%%%%%%%%%%%%%%%
%%% OR-LEFT
\begin{restatable}[$\OR$-left rule -- bottom-up]{proposition}{propshalloworl}
\label{prop:shalloworl}
    %\sonia{same as above. I am unsure about the status of thie proposition.}
	Let~$\J \in \set{\JIK, \JIKt, \JIKfour, \JISfour}$.
	Let~$\seq[\lseq]{\ }$ be an annotated LHS context where~$\seq[\lseq]{\inp{A}}, \oseq$ and~$\seq[\lseq]{\inp{B}}, \oseq$ are the premisses and~$\seq[\lseq]{\inp{A \OR B}}, \oseq$ is the conclusion sequent of an annotated instance of the $\orl$ rule. 
	Then, given realisation functions~$r_1$ on~$\seq[\lseq]{\inp{A}}$ and~$r_2$ on~$\seq[\lseq]{\inp{B}}$, there exists a substitution~$\sub$ with $\dom{\sub} \sset \negvar{\seq[\lseq]{\inp{A \OR B}}}$ and realisation function~$r$ on~$\seq[\lseq]{\inp{A \OR B}}$ such that
	$$
	\proves{\J}{\real{\seq[\lseq]{\inp{A \OR B}}} \IMPLIES (\subst{\real[r_1]{\seq[\lseq]{\inp{A}}}} \OR \subst{\real[r_2]{\seq[\lseq]{\inp{B}}}})}
	$$
\end{restatable}

\begin{proof}
	We proceed by induction on~$\depth{\seq[\lseq]{\ }}$.
	
	For the base case~$\depth{\seq[\lseq]{\ }} = 0$, we have $\seq[\lseq]{\ } = \lseq', \{ \}$ for some annotated LHS sequent~$\lseq'$.
	Note in the sequent~$\seq[\lseq]{\inp{A \OR B}}, \oseq = \lseq', \inp{A \OR B}, \oseq$, $A$ and $B$ are negative subformulas.
	Using Theorem~\ref{thm:merging} and Corollary~\ref{cor:nestedmerging} there exists a realisation~$r$ and substitution~$\sub$ such that $\real{\lseq'} \IMPLIES \subst{\real[r_i]{\lseq'}}$, $\real{A} \IMPLIES \subst{\real[r_1]{A}}$ and $\real{B} \IMPLIES \subst{\real[r_2]{B}}$ are theorems of $\J$. It follows using propositional reasoning that
	\begin{multline*}
		\proves{\J}
		{
			\real{(\lseq', \inp{A \OR B})} = (\real{\lseq'} \AND (\real{A} \OR \real{B}))
			\IMPLIES ((\real{\lseq'} \AND \real{A}) \OR (\real{\lseq'} \AND \real{B})) \\
			\IMPLIES (\subst{\real[r_1]{\lseq'}} \AND \subst{\real[r_1]{A}}) \OR (\subst{\real[r_2]{\lseq'}} \AND \subst{\real[r_2]{A}})
			= (\subst{\real[r_1]{(\lseq', A)}} \OR \subst{\real[r_2]{(\lseq', B)}})
		}
	\end{multline*}
	
	For the inductive case, $\seq[\lseq]{\ } = \lseq_0, \diabr[4i+3]{\seq[\lseq']{\ }}$ for some annotated LHS sequent~$\lseq_0$.
	By the inductive hypothesis, there exists a realisation~$r'$ on~$\seq[\lseq']{\inp{A \OR B}}$ and substitution~$\sub'$ with $\dom{\sub'} \sset \negvar{\seq[\lseq']{\inp{A \OR B}}}$ such that
	$$
		\proves{\J}
		{
			\real[r']{\seq[\lseq']{\inp{A \OR B}}} 
			\IMPLIES 
			(\subst[\sub']{\real[r_1]{\seq[\lseq']{\inp{A}}}} \OR \subst[\sub']{\real[r_2]{\seq[\lseq']{\inp{B}}}})
		}
	$$
	By Lemma~\ref{lem:facts}, $\restr[(\sub' \comp r_i)]{\lseq_0}$ is a realisation on~$\lseq_0$ for each~$i \in \set{1,2}$.
	By Corollary~\ref{cor:nestedmerging}, there exists a realisation~$r''$ on~$\lseq_0$ and substitution~$\sub''$ with $\dom{\sub''} \sset \negvar{\lseq_0}$ such that
	$$
		\proves{\J}
		{
			\real[r'']{\lseq_0} \IMPLIES \subst[\sub'']{\real[\sub' \comp r_i]{\lseq_0}} = \real[\sub'' \comp \sub' \comp r_i]{\lseq_0}
		}
	$$
	for each $i \in \set{1,2}$.
	By the Substitution Lemma~\ref{lem:subst}
	$$
	\proves{\J}
	{
		\real[\sub'' \comp r']{\seq[\lseq']{\inp{A \OR B}}} = \subst[\sub'']{\real[r']{\seq[\lseq']{\inp{A \OR B}}}} 
		\IMPLIES 
		(\subst[(\sub'' \comp \sub')]{\real[r_1]{\seq[\lseq']{\inp{A}}}} \OR \subst[(\sub \comp \sub')]{\real[r_2]{\seq[\lseq']{\inp{B}}}})
	}
	$$
	By the Lifting Lemma~\ref{lem:lifting}, there exists a satisfier term~$\mu(\svar[k])$ such that
	$$
	\proves{\J}
	{
		\sat{\svar[k]}{\real[\sub'' \comp r']{\seq[\lseq']{\inp{A \OR B}}}}
		\IMPLIES 
		\sat{\mu(\svar[k])}{(\subst[(\sub'' \comp \sub')]{\real[r_1]{\seq[\lseq']{\inp{A}}}} \OR \subst[(\sub \comp \sub')]{\real[r_2]{\seq[\lseq']{\inp{B}}}})}
	}
	$$
	Using the~$\jkax{3}$ axiom and transitivity
	$$
	\proves{\J}
	{
		\sat{\svar[k]}{\real[\sub'' \comp r']{\seq[\lseq']{\inp{A \OR B}}}}
		\IMPLIES 
		(\sat{\mu(\svar[k])}{\subst[(\sub'' \comp \sub')]{\real[r_1]{\seq[\lseq']{\inp{A}}}}}
		\OR
		\sat{\mu(\svar[k])}{\subst[(\sub'' \comp \sub')]{\real[r_2]{\seq[\lseq']{\inp{A}}}}})
	}
	$$
	Set $\sub'''$ as the substitution with $\svar[k] \mapsto \mu(\svar[k])$.
	With the diamond self-referential restriction,
	\begin{align*}
		\subst[(\sub''' \comp \sub'' \comp \sub')]{\real[r_1]{\diabr[4k+3]{\seq[\lseq']{\inp{A}}}}} &= \sat{\mu(\svar[k])}{\subst[(\sub'' \comp \sub')]{\real[r_1]{\seq[\lseq']{\inp{A}}}}}
		\\
		\subst[(\sub''' \comp \sub'' \comp \sub')]{\real[r_2]{\diabr[4k+3]{\seq[\lseq']{\inp{B}}}}} &= \sat{\mu(\svar[k])}{\subst[(\sub'' \comp \sub')]{\real[r_2]{\seq[\lseq']{\inp{B}}}}}
	\end{align*}
	The Substitution Lemma~\ref{lem:subst} gives
	$$
	\proves{\J}
	{
		\real[\sub''' \comp r'']{\lseq_0} = \subst[\sub''']{\real[r'']{\lseq_0}} \IMPLIES \subst[\sub''']{\real[\sub'' \comp \sub' \comp r_i]{\lseq_0}} = \real[\sub''' \comp \sub'' \comp \sub' \comp r_i]{\lseq_0}
	}
	$$
	Set~$\sub = \sub''' \comp \sub'' \comp \sub'$.
	Propositional reasoning gives
	\begin{multline*}
		\proves{\J}
		{
			(\real[\sub''' \comp r'']{\lseq_0} \AND \sat{\svar[k]}{\real[\sub'' \comp r']{\seq[\lseq']{\inp{A \OR B}}}})
			\\ \IMPLIES 
			((\subst[(\sub'' \comp \sub')]{\real[r_1]{\lseq_0}} \AND \sat{\mu(\svar[k])}{\subst[(\sub'' \comp \sub')]{\real[r_1]{\seq[\lseq']{\inp{A}}}}})
			\\ \OR
			(\subst[(\sub'' \comp \sub')]{\real[r_2]{\lseq_0}} \AND \real[\sub'' \comp \sub' \comp r_2]{\lseq_0} \AND \sat{\mu(\svar[k])}{\subst[(\sub'' \comp \sub')]{\real[r_2]{\seq[\lseq']{\inp{A}}}}}))
			\\
			= (\subst{\real[r_1]{(\lseq_0, \diabr[4k+3]{\seq[\lseq']{\inp{A}}})}} \OR \subst{\real[r_2]{(\lseq_0, \diabr[4k+3]{\seq[\lseq']{\inp{B}}})}})
		}
	\end{multline*}
	Set $r \colonequals \restr[(\sub''' \comp r'')]{\lseq_0} \UNI \restr[(\sub'' \comp r')]{\seq[\lseq']{\inp{A \OR B}}} \UNI \set{(4k+3, \svar[k])}$ which is a realisation function on~$\seq[\lseq]{\inp{A \OR B}}$ by Lemma~\ref{lem:facts}.
    \qed
\end{proof}

%%%%%%%%%%%%%%%%%%%%%%%%%%%%%%%%%%%%%%%%%%%%%%%%%%%%%%%%%%%%%%%%%%%%%%%%%%%%%%%%%%%%%%

\begin{restatable}[$\OR$-left rule -- top-down]{lemma}{lemorlrule}
\label{lem:orlrule}
	Let~$\J \in \set{\JIK, \JIKt, \JIKfour, \JISfour}$.
	Let~$A, B \in \langann$.
	Let~$\seq{\inp{A}}$ and~$\seq{\inp{B}}$ be the premisses and~$\seq{\inp{A \OR B}}$ be the conclusion sequents of an annotated instance of the~$\orl$ rule.
	Then, given realisation functions~$r_1$ and~$r_2$ on~$\seq{\inp{A}}$ and~$\seq{\inp{B}}$ respectively, there exists a realisation function~$r$ and a substitution~$\sub$ such that
	$$\proves{\J}{\subst{\real[r_1]{\seq{\inp{A}}}} \IMPLIES \subst{\real[r_2]{\seq{\inp{B}}}} \IMPLIES \real{\seq{\inp{A \OR B}}}}$$
\end{restatable}
\begin{proof}\setcounter{case}{0}
    Note that the rule is equivalent formulated as
    $$
    \vliinf{\orl}{}
    {
        \seq{\seq[\lseq]{\inp{A \OR B}}, \oseq}
    }
    {
        \seq{\seq[\lseq]{\inp{A}}, \oseq}
    }
    {
        \seq{\seq[\lseq]{\inp{B}}, \oseq}
    }
    $$
    for some annotated LHS context~$\seq[\lseq]{\ }$ and RHS sequent~$\oseq$.
    
	We proceed by induction on~$\depth{\seq{\ }}$.
	
	For the base case~$\depth{\seq{\ }}=0$, we have $\seq{\seq[\lseq]{\ }, \oseq } = \seq[\lseq]{\ }, \oseq$.
    %\sonia{same as above: should it not just be $\lseq, \seq[]{\ }, \oseq$? }
	%
	By Corollary~\ref{cor:nestedmerging}, there exists a realisation $r'$ on $\oseq$ and a substitution~$\sub'$ with~$\dom{\sub'} \sset \negvar{\oseq}$ such that
	$$
		\proves{\J}
		{
			\subst[\sub']{\real[r_1]{\oseq}}
			\IMPLIES
			\subst[\sub']{\real[r_2]{\oseq}}
			\IMPLIES
			\real[r']{\oseq}
		}
	$$
	By Lemma~\ref{lem:facts}, $\restr[(\sub' \comp r_1)]{\seq[\lseq]{\inp{A}}}$ and $\restr[(\sub' \comp r_2)]{\seq[\lseq]{\inp{B}}}$ are realisation functions.
	By Proposition~\ref{prop:shalloworl}, there exists a realisation~$r''$ on $\seq[\lseq]{\inp{A \OR B}}$ and a substitution~$\sub''$ with $\dom{\sub''} \sset \negvar{\seq[\lseq]{\inp{A \OR B}}}$ such that
	$$
		\proves{\J}
		{
			\real[r'']{\seq[\lseq]{\inp{A \OR B}}} \IMPLIES (\subst[\sub'']{\real[\sub' \comp r_1]{\seq[\lseq]{\inp{A}}}} \OR \subst[\sub'']{\real[\sub' \comp r_2]{\seq[\lseq]{\inp{B}}}})
		}
	$$
	By the Substitution Lemma~\ref{lem:subst},
	$$
	\proves{\J}
	{
		\subst[(\sub'' \comp \sub')]{\real[r_1]{\oseq}}
		\IMPLIES
		\subst[(\sub'' \comp \sub')]{\real[r_2]{\oseq}}
		\IMPLIES
		\subst[\sub'']{\real[r']{\oseq}}
	}
	$$
	Equivalently by Lemma~\ref{lem:facts}
	$$
	\proves{\J}
	{
		\subst[(\sub'' \comp \sub')]{\real[r_1]{\oseq}}
		\IMPLIES
		\subst[(\sub'' \comp \sub')]{\real[r_2]{\oseq}}
		\IMPLIES
		\real[\sub'' \comp r']{\oseq}
	}
	$$
	and by propositional reasoning
	\begin{multline*}
		\proves{\J}
		{
			(\subst[(\sub'' \comp \sub')]{\real[r_2]{\seq[\lseq]{\inp{A}}}} \IMPLIES  \subst[(\sub'' \comp \sub')]{\real[r_1]{\oseq}})
			\\ \IMPLIES
			(\subst[(\sub'' \comp \sub')]{\real[r_2]{\seq[\lseq]{\inp{B}}}} \IMPLIES \subst[(\sub'' \comp \sub')]{\real[r_2]{\oseq}})
			\\ \IMPLIES
			(\real[r'']{\seq[\lseq]{\inp{A \OR B}}} \IMPLIES \real[\sub'' \comp r']{\oseq})
		}
	\end{multline*}
	Set $r \colonequals r' \UNI (\sub'' \comp r')$ and $\sub \colonequals \sub'' \comp \sub'$.
	
	The inductive case is the same as the inductive case for $\andr$ in Lemma~\ref{lem:andrrule}.
    \qed
\end{proof}

%%%%%%%%%%%%%%%%%%%%%%%%%%%%%%%%%%%%%%%%%%%%%%%%%%%%%%%%%%%%%%%%%%%%%%%%%%%%%%%%%%%%%%
%%%%%%%%%%%%%%%%%%%%%%%%%%%%%%%%%%%%%%%%%%%%%%%%%%%%%%%%%%%%%%%%%%%%%%%%%%%%%%%%%%%%%%
%%% LEFT-RULES
\begin{restatable}[Other left rules -- bottom-up]{proposition}{propshallowlrules}\label{prop:shallowlrules}
	Let~$\J \in \set{\JIK, \JIKt, \JIKfour, \JISfour}$.
	Let~$\seq[\lseq]{\ }$ be an annotated LHS context where~$\seq[\lseq]{\Omega'}, \oseq$ is the premiss and~$\seq[\lseq]{\Omega}, \oseq$ is the conclusion sequent of an annotated instance of a rule~$\rle \in \set{\andl, \boxl, \dial, \cont, \tl, \fourl}$ where~$\Omega'$ and~$\Omega$ are LHS sequents. 
	Then, given realisation function~$r'$ on~$\seq[\lseq]{\Omega'}$, there exists a substitution~$\sub$ and realisation function~$r$ on~$\seq[\lseq]{\Omega}$ such that:
    $\proves{\J}{\real{\seq[\lseq]{\Omega}} \IMPLIES \subst{\real[r']{\seq[\lseq]{\Omega'}} }}$
    with $\J$ appropriately containing the justification axioms corresponding to the modal rules considered.
%	\begin{itemize}
%		\item $\proves{\J}{\real{\seq[\lseq]{\Omega}} \IMPLIES \subst{\real[r']{\seq[\lseq]{\Omega'}} }}$ if~$\rle \in \set{\andl, \boxl, \dial, \cont}$;
%		\item $\proves{\ILP}{\real{\seq[\lseq]{\Omega}} \IMPLIES \subst{\real[r']{\seq[\lseq]{\Omega'}} }}$ if~$\rle \in \set{\tl, \fourl}$.
%	\end{itemize}
\end{restatable}
\begin{proof}\setcounter{case}{0}
	We proceed by induction on~$\depth{\seq[\lseq]{\ }}$.
	For the base case~$\depth{\seq[\lseq]{\ }} = 0$, we have $\seq[\lseq]{\ } = \lseq_0, \{ \ \}$ for some annotated LHS sequent $\lseq_0$.
	We consider the following cases:
	\begin{case}
		When $\rle = \boxl$.
		$\Omega = \lseq_0, \inp{\BOX[4j+1] A}, \diabr[4i+3]{\lseq_1}$ and $\Omega' = \lseq_0, \diabr[4i+3]{\inp{A}, \lseq_1}$.
		First note
		$$\proves{\J}
		{
			\real[r']{A} \IMPLIES \real[r']{\lseq_1}
			\IMPLIES
			(\real[r']{A} \AND \real[r']{\lseq_1})
		}
		$$
		By the Lifting Lemma~\ref{lem:lifting}, there exists a proof term~$t(\pvar[j])$ such that
		$$\proves{\J}
		{
			\just{\pvar[j]}{A}
			\IMPLIES
			\just{t(\pvar[j])}{
				(\real[r']{\lseq_1}
				\IMPLIES
				(\real[r']{A} \AND \real[r']{\lseq_1}))
			}
		}		
		$$
		Using the $\jkax{2}$ axiom and transitivity, we achieve
		$$\proves{\J}
		{
			\just{\pvar[j]}{\real[r']{A}}
			\IMPLIES
			\sat{\svar[i]}{\real[r']{\lseq_1}}
			\IMPLIES
			\sat{(\sappl{t(\pvar[j])}{\svar[i]})}{(\real[r']{A} \AND \real[r']{\lseq_1})}
		}
		$$
		By propositional reasoning
		$$\proves{\J}
		{
			(\just{\pvar[j]}{\real[r']{A}}
			\AND
			\sat{\svar[i]}{\real[r']{\lseq_1}})
			\IMPLIES
			\sat{(\sappl{t(\pvar[j])}{\svar[i]})}{(\real[r']{A} \AND \real[r']{\lseq_1})}
		}
		$$
		Set a substitution~$\sub$ with~$\subst{\svar[i]} \colonequals \sappl{t(\pvar[j])}{\svar[i]}$.
		Set~$r \colonequals \restr[(\sub \comp r')]{\lseq_0} \UNI \restr[r']{A} \UNI \restr[r']{\lseq_1} \UNI \set{(4j+1, \pvar[j]), (4i+3, \svar[i])}$ which is a realisation function on~$\Omega$ by Lemma~\ref{lem:facts}.
		We have
		$$\proves{\J}
		{
			\real{\Omega} = 
			(\real{\lseq_0}
			\AND
			\just{\pvar[j]}{\real{A}}
			\AND
			\sat{\svar[i]}{\real{\lseq_1}})
			\IMPLIES
			(\real[\sub \comp r']{\lseq_0}
			\AND
			\sat{(\sappl{t(\pvar[j])}{\svar[i]})}{(\real[r']{A} \AND \real[r']{\lseq_1})})
			= \subst{\real[r']{\Omega'}}
		}		
		$$
	\end{case}
	\begin{case}
		When~$\rle \in \set{\andl, \dial}$.
		We note $\real[r']{\Omega} = \real[r']{\Omega}$ so set~$r \colonequals r'$ and $\sub \colonequals \Idnty$.
	\end{case}
        For convenience in the further cases, we assume~$\seq[\lseq]{\ } = \{ \ \}$.
	\begin{case}
		When~$\rle = \cont$.
		$\Omega = \lseq_3$, and $\Omega' = \lseq_1, \lseq_2$ for distinct annotated LHS sequents~$\lseq_1, \lseq_2, \lseq_3$ for the same unannotated LHS sequent~$\lseq$.
		Set~$X = \set[(i,j,k)]{\text{Annotations $i,j,k$ occur in the same position of $\lseq_1, \lseq_2, \lseq_3$ respectively}}$.
		Note as $\seq[\lseq]{\lseq_1, \lseq_2}, \oseq$ and~ $\seq[\lseq]{\lseq_3}, \oseq$ are properly annotated, $(i, j, k) \in X$ only if $i,j,k$ are all odd or all even.
		For each: $(4i+1, 4j+1, 4k+1) \in X$, set $\subst[\sub']{\pvar[i]} = \subst[\sub']{\pvar[j]} \colonequals \pvar[k]$; $(4i+3, 4j+3, 4k+3) \in X$, set $\subst[\sub']{\svar[i]} = \subst[\sub']{\svar[j]} \colonequals \svar[k]$.
		
		Now, for each $(i, j, k) \in X$, set~$r'_1(k) = \sub' \comp r'(i)$ and $r'_2(k) = \sub' \comp r'(j)$. 
		We have two realisations $r'_1$ and $r'_2$ on $\lseq_3$ with $\real[r'_1]{\lseq_3} = \real[\sub' \comp r]{\lseq_1}$ and $\real[r'_2]{\lseq_3} = \real[\sub' \comp r]{\lseq_2}$.
		Using Corollary~\ref{cor:nestedmerging}, there exists a realisation~$\hat{r}$ on~$\lseq_3$ and a substitution~$\hat{\sub}$ with $\dom{\hat{\sub}} \sset \negvar{\lseq_3}$ and the no new variable condition such that $\real[\hat{r}]{\lseq_3} \IMPLIES \subst[\hat{\sub}]{ \real[r'_1]{\lseq_3}} = \real[\hat{\sub} \comp \sub' \comp r]{A_1}$ and $\real[\hat{r}]{\lseq_3} \IMPLIES \subst[\hat{\sub}]{ \real[r'_2]{\lseq_3}} = \real[\hat{\sub} \comp \sub' \comp r]{\lseq_2}$ are theorems of~$\J$.
		Set $r \colonequals \hat{r} \UNI \restr[r']{\lseq_0}$ which is a realisation function on $\Omega$ by Lemma~\ref{lem:facts}.
		Set  $\sub = \hat{\sub} \comp \sub'$. 
		By propostional reasoning we achieve
		$$\proves{\J}{\real{{\lseq_3}} \IMPLIES \subst{\real[r']{({\lseq_1}, {\lseq_2})}}}$$
	\end{case}
	\begin{case}
		When $\rle = \tl$.
		$\Omega = \inp{\BOX[4i+1] A}$ and $\Omega' = \inp{A}$.
		Set $r \colonequals r' \UNI \set{(4i+1, \pvar[i])}$ which is a realisation function on~$\Omega$ by Lemma~\ref{lem:facts} and~$\sub \colonequals \id$ and we have
		$$\proves{\ILP}
		{
			\real{\Omega} = \just{\pvar[i]}{\real{A}}
			\IMPLIES
			\real[r']{A} = \subst{\real[r']{\Omega'}}
		}
		$$
		using the $\jtaxb$ axiom.
	\end{case}
	\begin{case}
		When~$\rle = \fourl$.
		$\Omega = \inp{\BOX[4k+1] A}, \diabr[4i+3]{\lseq_1}$ and $\Omega' = \diabr[4i+3]{\inp{\BOX[4k+1] A, \lseq_1}}$.
		Note
		$$\proves{\ILP}
		{
			\just{\pvar[k]}{\real[r']{A}}
			\IMPLIES
			\real[r']{\lseq_1}
			\IMPLIES
			(
			\just{\pvar[k]}{\real[r']{A}}
			\AND
			\real[r']{\lseq_1}
			)
		}
		$$
		By the Lifting Lemma~\ref{lem:lifting}, there exists a proof term~$t(\jbang{\pvar[k]})$ such that
		$$\proves{\ILP}
		{
			\just{\jbang{\pvar[k]}}{\just{\pvar[k]}{\real[r']{A}}}
			\IMPLIES
			\just
			{t(\jbang{\pvar[k]})}
			{(\real[r']{\lseq_1}
				\IMPLIES
				(
				\just{\pvar[k]}{\real[r']{A}}
				\AND
				\real[r']{\lseq_1}
				))}
		}
		$$
		Using the $\jkax{2}$ axiom and transitivity, we achieve
		$$\proves{\ILP}
		{
			\just{\jbang{\pvar[k]}}{\just{\pvar[k]}{\real[r']{A}}}
			\IMPLIES
			\sat{\svar[i]}{\real[r']{\lseq_1}}
			\IMPLIES
			(
			\sat{(\sappl{t(\jbang{\pvar[k]})}{\svar[i]})}
			{
				\just{\pvar[k]}{\real[r']{A}}
				\AND
				\real[r']{\lseq_1}
			}
			)
		}
		$$
		Using the $\jfaxb$ axiom~$\just{\pvar[i]}{A} \IMPLIES \just{\jbang{\pvar[i]}}{\just{\pvar[i]}{A}}$ and transitivity
		$$\proves{\ILP}
		{
			\just{\pvar[k]}{\real[r']{A}}
			\IMPLIES
			\sat{\svar[i]}{\real[r']{\lseq_1}}
			\IMPLIES
			(
			\sat{(\sappl{t(\jbang{\pvar[k]})}{\svar[i]})}
			{
				\just{\pvar[k]}{\real[r']{A}}
				\AND
				\real[r']{\lseq_1}
			}
			)
		}$$
		By propositional reasoning
		$$\proves{\ILP}
		{
			(\just{\pvar[k]}{\real[r']{A}}
			\AND
			\sat{\svar[i]}{\real[r']{\lseq_1}})
			\IMPLIES
			(
			\sat{(\sappl{t(\jbang{\pvar[k]})}{\svar[i]})}
			{
				\just{\pvar[k]}{\real[r']{A}}
				\AND
				\real[r']{\lseq_1}
			}
			)
		}$$
		Set~$r \colonequals r'$ and $\sub$ with $\subst{\svar[i]} \colonequals \sappl{t(\jbang{\pvar[k]})}{\svar[i]}$, to achieve
		$$\proves{\ILP}
		{	
			\real{\Omega} =
			(\just{\pvar[k]}{\real[r']{A}}
			\AND
			\sat{\svar[i]}{\real[r']{\lseq_1}})
			\IMPLIES
			(
			\sat{(\sappl{t(\jbang{\pvar[k]})}{\svar[i]})}
			{
				\just{\pvar[k]}{\real[r']{A}}
				\AND
				\real[r']{\lseq_1}
			}
			)
			= \subst{\real[r']{\Omega'}}
		}$$
	\end{case}

	For the inductive case~$\seq[\lseq]{\ } = \lseq_0, \diabr[4k+3]{\lseq_1, \seq[\lseq_2]{\ }}$.
	By the inductive hypothesis. there exists a realisation~$\hat{r}$ and substitution~$\sub'$ with $\dom{\sub'} \sset \negvar{\seq[\lseq_2]{\Omega'}}$ such that
	$$\proves{\J}{\real[\hat{r}]{\seq[\lseq_2]{\Omega}} \IMPLIES \subst{\real[r']{\seq[\lseq_2]{\Omega'}}}}$$
	By propositional reasoning
	$$\proves{\J}{(\real[\sub' \comp r']{\lseq_1} \AND  \real[\hat{r}]{\seq[\lseq_2]{\Omega}}) \IMPLIES (\subst[\sub']{\real[r']{\lseq_1}} \AND \subst[\sub']{\real[r']{\seq[\lseq_2]{\Omega'}}})}$$
	By the Lifting Lemma~\ref{lem:lifting}, there exists a satisfier term~$\mu(\svar[k])$ such that
	$$\proves{\J}{\sat{\svar[k]}{(\real[\sub' \comp r]{\lseq_1} \AND \real[\hat{r}]{\seq[\lseq_2]{\Omega}})} \IMPLIES \sat{\mu(\svar[k])}{(\subst[\sub']{\real{\lseq_1}} \AND \subst[\sub']{\real{\seq[\lseq_2]{\Omega'}}})}}$$
	Set $\sub''$ to be the substitution with~$\svar[k] \mapsto \mu(\svar[k])$.
	Set $\sub \colonequals \sub'' \comp \sub'$.
	With the satisfier self-referentiality restriction, we have
	$$
		\sat{(\sub \comp r')(4k+3)}{(\subst{\real{\lseq_1}} \AND \subst{\real{\seq[\lseq_2]{\Omega'}}})}
		=
		\sat{\mu(\svar[k])}{(\subst[\sub']{\real{\lseq_1}} \AND \subst[\sub']{\real{\seq[\lseq_2]{\Omega'}}})}
	$$
	and further propositional reasoning gives
	$$\proves{\J}{(\real[\sub \comp r']{\lseq_0} \AND \sat{\svar[k]}{(\real[\sub' \comp r']{\lseq_1} \AND \real[\hat{r}]{\seq[\lseq_2]{\Omega}})}) 
		\IMPLIES 
	(\real[\sub \comp r']{\lseq_0} \AND \sat{(\sub \comp r')(4k+3)}{(\subst{\real[r']{\lseq_1}} \AND \subst{\real[r']{\seq[\lseq_2]{\Omega'}}})})}$$
	Set $r \colonequals \restr[(\sub \circ r')]{\lseq_0} \UNI \restr[(\sub' \circ r')]{\lseq_1} \UNI \hat{r} \UNI \set{(4k+3, \svar[k])}$.
	By Lemma~\ref{lem:facts}, this is a realisation function on~$\seq[\lseq]{\Omega }$, and we can conclude
	$$\proves{\J}{\real{\seq[\lseq]{\Omega }} \IMPLIES \subst{\real[r']{\seq[\lseq]{\Omega '}}}}$$
    \qed
\end{proof}
%%%%%%%%%%%%%%%%%%%%%%%%%%%%%%%%%%%%%%%%%%%%%%%%%%%%%%%%%%%%%%%%%%%%%%%%%%%%%%%%%%%%%%
\begin{restatable}[Other left rules -- top-down]{lemma}{lemleftrules}
\label{lem:leftrules}
	Let $\J \in \set{\JIK, \JIKt, \JIKfour, \JISfour}$.
	Let~$\seq{\ }$ be an annotated context where~$\seq{\Omega'}$ is the premiss and~$\seq{\Omega}$ is the conclusion sequent of an annotated instance of a rule~$\rle \in \set{\andl, \boxl, \dial, \cont, \tl, \fourl}$. 
	Then, given realisation function~$r'$ on~$\seq{\Omega'}$, there exists a realisation function~$r$ on~$\fseq$ and a substitution~$\sub$ such that:
    $\proves{\J}{\subst{\real[r']{\seq{\Omega'}}} \IMPLIES \real{\seq{\Omega}}}$
    with $\J$ adequately containing the justification axioms corresponding to the modal rule considered.
%	\begin{itemize}
%		\item $\proves{\J}{\subst{\real[r']{\seq{\Omega'}}} \IMPLIES \real{\seq{\Omega}}}$ if~$\rle \in \set{\andl, \boxl, \dial, \cont}$;
%		\item $\proves{\ILP}{\subst{\real[r']{\seq{\Omega'}}} \IMPLIES \real{\seq{\Omega}}}$ if~$\rle \in \set{\tl, \fourl}$.
%	\end{itemize}
\end{restatable}
\begin{proof}
    Note that a rule instance is of the form
    $$
    \vlinf{\rle}{}
    {
        \seq{\seq[\lseq]{\Omega'}, \oseq}
    }
    {
        \seq{\seq[\lseq]{\Omega}, \oseq}
    }
    $$
    
	We proceed by induction on $\depth{\seq{\ }}$.
	
	For the base case $\depth{\seq{\ }}=0$, we have $\seq[\lseq]{\Omega'}, \oseq$ as the premiss sequent and $\seq[\lseq]{\Omega}, \oseq$ as the conclusion sequent.
	By Proposition~\ref{prop:shallowlrules}, there exists a realisation $r''$ on $\seq[\lseq]{\Omega}$ and and a substitution $\sub$ with $\dom{\sub} \sset \negvar{\seq[\lseq]{\Omega}}$ such that
	$$
	\proves{\J}
	{
		\real[r'']{\seq[\lseq]{\Omega}} \IMPLIES \subst{\real[r']{\seq[\lseq]{\Omega'}}}
	}
	$$
	By propositional reasoning
	$$
	\proves{\J}
	{
		(\subst{\real[r']{\seq[\lseq]{\Omega'}}} \IMPLIES \subst{\real[r']{\oseq}})
		\IMPLIES
		(\real[r'']{\seq[\lseq]{\Omega}} \IMPLIES \real[\sub \comp r']{\oseq})
	}
	$$
	Set $r \colonequals r'' \UNI \restr[(\sub \comp r')]{\oseq}$ which is a realisation function on~$\seq[\lseq]{\Omega}, \oseq$ by Lemma~\ref{lem:facts}.
	
	The inductive case is similar to Lemma~\ref{lem:andrrule} and is omitted.
    \qed
\end{proof}
%%%%%%%%%%%%%%%%%%%%%%%%%%%%%%%%%%%%%%%%%%%%%%%%%%%%%%%%%%%%%%%%%%%%%%%%%%%%%%%%%%%%%%
%%%%%%%%%%%%%%%%%%%%%%%%%%%%%%%%%%%%%%%%%%%%%%%%%%%%%%%%%%%%%%%%%%%%%%%%%%%%%%%%%%%%%%
%%% IMP-LEFT
\begin{restatable}[$\IMP$-left rule -- bottom-up]{proposition}{propshallowcimpl}
\label{prop:shallowcimpl}
    Let $\J \in \set{\JIK, \JIKt, \JIKfour, \JISfour}$,
    Given a shallow annotated rule instance of the $\constrimpl$ rule:
	$$
	\vliinf{\constrimpl}{}
	{
		\seq[\Omega]{\inp{A \IMPLIES B}}, \lseq, \oseq
	}
	{
		\seq[\Omega^{\br{}}]{\out{A}}, \lseq
	}
	{
		\seq[\Omega]{\inp{B}}, \oseq
	}
	$$
	with realisation functions $r_1$ on $\seq[\Omega^{\br{}}]{\out{A}}, \lseq$ and $r_2$ on $\seq[\Omega]{\inp{B}}, \oseq$, there exists substitutions $\sub_1, \sub_2$ and realisation function $r$ on $\seq[\Omega]{\inp{A \IMPLIES B}}, \lseq, \oseq$ such that
	$$
	\proves{\J}
	{
		\subst[\sub_1]{\real[r_1]{(\seq[\Omega^{\br{}}]{\out{A}}, \lseq)}}
		\IMPLIES
		\subst[\sub_2]{\real[r_2]{(\seq[\Omega]{\inp{B}}, \oseq)}}
		\IMPLIES
		\real{(\seq[\Omega]{\inp{A \IMPLIES B}}, \lseq, \oseq)}
	}
	$$
\end{restatable}
\begin{proof}
	We first show, the following: given realisations $r_1$ on $\seq[\Omega^{\br{}}]{\out{A}}$ and $r_2$ on $\seq[\Omega]{\inp{B}}$, there exists substitutions $\sub_1, \sub_2$ with $\dom{\sub_1} \sset \negvar{\seq[\Omega^{\br{}}]{\out{A}}}$ and $\dom{\sub_2} \sset \negvar{\real[r_2]{\seq[\Omega]{\inp{B}}}}$, and realisation $r$ such that
	$$
	\proves{\J}
	{
		\subst[\sub_1]{\real[r_1]{\seq[\Omega^{\br{}}]{\out{A}}}}
		\IMPLIES
		\real{\seq[\Omega]{\inp{A \IMPLIES B}}}
		\IMPLIES
		\subst[\sub_2]{\real[r_2]{\seq[\Omega]{\inp{B}}}}
	}
	$$
	by induction on $\depth{\seq[\Omega]{\ }}$.
	
	For the base case $\depth{\seq[\Omega]{\ }} = 0$, we have $\seq[\Omega]{\ } = \seq[\Omega^{\br{}}]{\ } = \lseq', \{\ \}$.
	Extend $r_2$ so that $\real[r_2]{A} = \real[r_1]{A}$.
	We have by Corollary~\ref{cor:nestedmerging} a realisation $r'$ and a substitution~$\sub'$ with
	\begin{align*}
		\proves{\J}{& \real[r']{\lseq'} \IMPLIES \subst[\sub']{\real[r_i]{\lseq'}}} \\
		\proves{\J}{& \real[r']{B} \IMPLIES \subst[\sub']{\real[r_i]{B}}} \\
		\proves{\J}{& \subst[\sub']{\real[r_1]{A}} \IMPLIES \real[r_1]{A}}
	\end{align*}
	By propositional reasoning we have
	$$
	\proves{\J}
	{
		(\subst[\sub']{\real[r_1]{\lseq'}} \IMPLIES \subst[\sub']{\real[r_1]{A}})
		\IMPLIES
		(\real[r']{\lseq'} \AND \real[r']{(A \IMPLIES B)})
		\IMPLIES
		(\subst[\sub']{\real[r_1]{\lseq'}} \IMPLIES \subst[\sub']{\real[r_1]{A}})
	}
	$$
	We set $r \colonequals r'$, $\sub_1 \colonequals \sub'$ and $\sub_2 \colonequals \sub'$.
	
	For the inductive case, we have $\seq[\Omega^{\br{}}]{\out{A}} = \lseq', \br[4n]{\seq[\Omega'^{\br{}}]{\out{A}}}$ and $\seq[\Omega]{\inp{B}} = \lseq', \diabr[4k+3]{\seq[\Omega']{\inp{B}}}$.
	By the inductive hypothesis, there exists a realisation~$r'$ and substitutions~$\sub_1'$, $\sub_2'$ with $\dom{\sub_1'} \sset \negvar{\seq[\Omega'^{\br{}}]{\out{A}}}$ and $\dom{\sub_2'} \sset \negvar{\seq[\Omega']{\inp{B}}}$ such that
	$$
		\proves{\J}
		{
			\subst[\sub_1']{\real[r_1]{\seq[\Omega'^{\br{}}]{\out{A}}}}
			\IMPLIES
			\real[r']{\seq[\Omega']{\inp{A \IMPLIES B}}}
			\IMPLIES
			\subst[\sub_2']{\real[r_2]{\seq[\Omega']{\inp{B}}}}
		}
	$$
	By Lemma~\ref{lem:lifting}, there exists a satisfier term $\mu \equiv \mu((\sub_1' \comp r_1)(4n), \svar[k])$ such that
	$$
	\proves{\J}
	{
		\just{(\sub_1' \comp r_1)(4n)}{(\subst[\sub_1']{\real[r_1]{\seq[\Omega'^{\br{}}]{\out{A}}}})}
		\IMPLIES
		\sat{\svar[k]}{\real[r']{\seq[\Omega']{\inp{A \IMPLIES B}}}}
		\IMPLIES
		\sat{\mu}{(\subst[\sub_2']{\real[r_2]{\seq[\Omega']{\inp{B}}}})}
	}
	$$
	Set $\sub_2''$ as the extension of $\sub_2'$ with $\svar[k] \mapsto \mu$.
	With the satisfier self-referentiality restriction, we have
	$$
	\sat{\mu}{(\subst[\sub_2']{\real[r_2]{\seq[\Omega']{\inp{B}}}})} = \sat{(\sub_2'' \comp r_2)(4k+3)}{(\subst[\sub_2'']{\real[r_2]{\seq[\Omega']{\inp{B}}}})}
	$$
	Now note when restricted to $\lseq'$, $(\sub_1' \comp r_1)$ and $(\sub_2'' \comp r_2)$ are realisations.
	By Theorem~\ref{thm:merging}, there exists a realisation $r''$ and substitution $\sub'$ with $\dom{\sub'} \sset \negvar{\lseq'}$ such that
	\begin{align*}
		\proves{\J}
		{&
			\real[r'']{\lseq'}
			\IMPLIES
			\subst[\sub']
			{
				\real[(\sub_1' \comp r_1)]{\lseq'}
			}
			= 
			\subst[(\sub' \comp \sub_1')]
			{
				\real[r_1]{\lseq'}
			}
		}
		\\
		\proves{\J}
		{&
			\real[r'']{\lseq'}
			\IMPLIES
			\subst[\sub']
			{
				\real[(\sub_2'' \comp r_2)]{\lseq'}
			}
			=
			\subst[(\sub' \comp \sub_2'')]
			{
				\real[r_2]{\lseq'}
			}
		}
	\end{align*}
	By the Substitution Lemma~\ref{lem:subst} and Lemma~\ref{lem:facts}
	% $$
	% \proves{\J}
	% {
	% 	\just{(\sub' \comp \sub_1' \comp r_1)(4n)}{(\subst[(\sub' \comp \sub_1')]{\real[r_1]{\seq[\Omega'^{\br{}}]{\out{A}}}})}
	% 	\IMPLIES
	% 	\sat{\svar[k]}{\real[(\sub' \comp r')]{\seq[\Omega']{\inp{A \IMPLIES B}}}}
	% 	\IMPLIES
	% 	\sat{\mu}{(\subst[(\sub' \comp \sub_2'')]{\real[r_2]{\seq[\Omega']{\inp{B}}}})}
	% }
	% $$
        \begin{multline*}
            \proves{\J}
	{
		\just{(\sub' \comp \sub_1' \comp r_1)(4n)}{(\subst[(\sub' \comp \sub_1')]{\real[r_1]{\seq[\Omega'^{\br{}}]{\out{A}}}})} \\
		\IMPLIES
		\sat{\svar[k]}{\real[(\sub' \comp r')]{\seq[\Omega']{\inp{A \IMPLIES B}}}}
		\IMPLIES
		\sat{\mu}{(\subst[(\sub' \comp \sub_2'')]{\real[r_2]{\seq[\Omega']{\inp{B}}}})}
	}
        \end{multline*}
	Set $\sub_1 \colonequals \sub' \comp \sub_1'$ and $\sub_2 \colonequals \sub' \comp \sub_2''$. 
	Equivalently,
	$$
	\proves{\J}
	{
		\subst[\sub_1]{\real[r_1]{\br[4n]{\seq[\Omega'^{\br{}}]{\out{A}}}}}
		\IMPLIES
		\sat{\pvar[k]}{\real[(\sub' \comp r')]{\seq[\Omega']{\inp{A \IMPLIES B}}}}
		\IMPLIES
		\subst[ \sub_2]{\real[r_2]{\diabr[4k+3]{\seq[\Omega']{\inp{B}}}}}
	}
	$$
	By propositional reasoning
	\begin{multline*}
		\proves{\J}
		{
			\subst[\sub_1]{\real[r_1]{(\lseq', \br[4n]{\seq[\Omega'^{\br{}}]{\out{A}}})}} = (\subst[\sub_1]{\real[r_1]{\lseq'}} \IMPLIES \subst[\sub_1]{\real[r_1]{\br[4n]{\seq[\Omega'^{\br{}}]{\out{A}}}}})
			\\
			\IMPLIES
			(\real[r'']{\lseq'} \AND \sat{\pvar[k]}{\real[(\sub' \comp r')]{\seq[\Omega']{\inp{A \IMPLIES B}}}})
			\\
			\IMPLIES
			(\subst[\sub_2]{\real[r_2]{\lseq'}} \AND \subst[ \sub_2]{\real[r_2]{\diabr[4k+3]{\seq[\Omega']{\inp{B}}}}})
			=
			\subst[\sub_2]{\real[r_2]{(\lseq', \diabr[4k+3]{\seq[\Omega']{\inp{B}}})}}
		}
	\end{multline*}
	Set $r \colonequals r'' \UNI (\sub' \comp r') \UNI \set{(4k+3, \svar[k])}$ which is a realisation function on~$\seq[\Omega]{\inp{A \IMPLIES B}}$ by Lemma~\ref{lem:facts}.
	
	Now, returning to the original problem with given realisations $r_1$ on $\seq[\Omega^{\br{}}]{\out{A}}, \lseq$ and $r_2$ on $\seq[\Omega]{\inp{B}}, \oseq$.
	We have a realisation~$r'$ on $\seq[\Omega]{\inp{A \IMPLIES B}}$ and substitutions $\sub_1$ and $\sub_2$ with $\dom{\sub_1} \sset \negvar{\seq[\Omega^{\br{}}]{\out{A}}}$ and $\dom{\sub_2} \sset \negvar{\seq[\Omega]{\inp{B}}}$ such that
	$$
	\proves{\J}
	{
		\subst[\sub_1]{\real[r_1]{\seq[\Omega^{\br{}}]{\out{A}}}}
		\IMPLIES
		\real[r']{\seq[\Omega]{\inp{A \IMPLIES B}}}
		\IMPLIES
		\subst[\sub_2]{\real[r_2]{\seq[\Omega]{\inp{B}}}}
	}
	$$
	By propositional reasoning
	$$
	\proves{\J}
	{
		(\subst[\sub_1]{\real[r_1]{\lseq}} \IMPLIES \subst[\sub_1]{\real[r_1]{\seq[\Omega^{\br{}}]{\out{A}}}})
		\IMPLIES
		(\real[r']{\seq[\Omega]{\inp{A \IMPLIES B}}} \AND \real[(\sub_1 \comp r_1)]{\lseq})
		\IMPLIES
		\subst[\sub_2]{\real[r_2]{\seq[\Omega]{\inp{B}}}}
	}
	$$
	and then
	$$
	\proves{\J}
	{
		\subst[\sub_1]{\real[r_1]{(\seq[\Omega^{\br{}}]{\out{A}}, \lseq)}}
		\IMPLIES
		\subst[\sub_2]{\real[r_2]{(\seq[\Omega]{\inp{B}}, \oseq)}}
		\IMPLIES
		((\real[r']{\seq[\Omega]{\inp{A \IMPLIES B}}} \AND \real[(\sub_1 \comp r_1)]{\lseq}) \IMPLIES \real[(\sub_2 \comp r_2)]{\oseq})
	}
	$$
	Set $r \colonequals r' \UNI \restr[(\sub_1 \comp r_1)]{\lseq} \UNI \restr[(\sub_2 \comp r_2)]{\oseq}$ which is a realisation function on~$\seq[\Omega]{\inp{A \IMPLIES B}}, \lseq, \oseq$ by Lemma 4.
    \qed
\end{proof}
%%%%%%%%%%%%%%%%%%%%%%%%%%%%%%%%%%%%%%%%%%%%%%%%%%%%%%%%%%%%%%%%%%%%%%%%%%%%%%%%%%%%%%
%%% IMP-LEFT
\begin{restatable}[\todo{}tilde $\IMP$-left rule]{lemma}{lemimplrule}
\label{lem:implrule}
    %\sonia{same issue as above}
    Let $\J \in \set{\JIK, \JIKt, \JIKfour, \JISfour}$.
    Given an annotated rule instance of the $\constrimpl$ rule:
	$$
	\vliinf{\constrimpl}{}
	{
		\seq{\seq[\Omega]{\inp{A \IMPLIES B}}, \lseq, \oseq}
	}
	{
		\seq{\seq[\Omega^{\br{}}]{\out{A}}, \lseq}
	}
	{
		\seq{\seq[\Omega]{\inp{B}}, \oseq}
	}
	$$
	with realisation functions $r_1$ on $\seq{\seq[\Omega^{\br{}}]{\out{A}}, \lseq}$ and $r_2$ on $\seq{\seq[\Omega]{\inp{B}}, \oseq}$, there exists substitutions $\sub_1, \sub_2$ with $\dom{\sub_1} \sset \negvar{\seq{\seq[\Omega^{\br{}}]{\out{A}}, \lseq}}$ and $\dom{\sub_1} \sset \negvar{\seq{\seq[\Omega]{\inp{B}}, \oseq}}$, and realisation function $r$ on~$\seq{\seq[\Omega]{\inp{A \IMPLIES B}}, \lseq, \oseq}$ such that
	$$
	\proves{\J}
	{
		\subst[\sub_1]{\real[r_1]{\seq{\seq[\Omega^{\br{}}]{\out{A}}, \lseq}}}
		\IMPLIES
		\subst[\sub_2]{\real[r_2]{\seq{\seq[\Omega]{\inp{B}}, \oseq}}}
		\IMPLIES
		\real{\seq{\seq[\Omega]{\inp{A \IMPLIES B}}, \lseq, \oseq}}
	}
	$$
\end{restatable}
\begin{proof}
    We proceed by induction on~$\depth{\seq{\ }}$.
	
	The base case is covered by Proposition~\ref{prop:shallowcimpl}.

	For the inductive case~$\seq{\ } = \lseq, \br[4n]{\seq[\fseq']{\ }}$ for some annotated LHS sequent~$\lseq$ and context~$\seq[\fseq']{\ }$.
	Using the inductive hypothesis, there exists a realisation $r'$ on $\seq[\fseq']{\seq[\Omega]{\inp{A \IMPLIES B}}, \lseq, \oseq}$ and substitutions~$\sub_1'$ and $\sub_2'$ with~$\dom{\sub_1'} \sset \negvar{\seq[\fseq']{\seq[\Omega^{\br{}}]{\out{A}}, \lseq}}$ and $\dom{\sub_2'} \sset \negvar{\seq[\fseq']{\seq[\Omega]{\inp{B}}, \oseq}}$ such that
	$$\proves{\J}{\subst[\sub_1']{\real[r_1]{\seq[\fseq']{\seq[\Omega^{\br{}}]{\out{A}}, \lseq}}} \IMPLIES \subst[\sub_2']{\real[r_2]{\seq[\fseq']{\seq[\Omega]{\inp{B}}, \oseq}}} \IMPLIES \real[r']{\seq[\fseq']{\seq[\Omega]{\inp{A \IMPLIES B}}, \lseq, \oseq}}}$$
	By Lemma~\ref{lem:facts}, $\restr[(\sub_i' \comp r_i)]{\lseq}$ is a realisation on~$\lseq$ for each~$i \in \set{1,2}$.
	By Corollary~\ref{cor:nestedmerging}, there exists a realisation~$r''$ on~$\lseq$ and substitution~$\sub''$ with~$\dom{\sub''} \sset \negvar{\lseq}$ such that
	$$\proves{\J}{\real[r'']{\lseq} \IMPLIES \subst[\sub'']{\real[\sub_i' \comp r_i]{\lseq}}}$$
	for~$i \in \set{1, 2}$.
	By the Substitution Lemma~\ref{lem:subst}
        \begin{multline*}
            \proves{\J}{\subst[(\sub'' \comp \sub_1')]{\real[r_1]{\seq[\fseq']{\seq[\Omega^{\br{}}]{\out{A}}, \lseq}}} 
            \\
            \IMPLIES \subst[(\sub'' \comp \sub_2')]{\real[r_2]{\seq[\fseq']{\seq[\Omega]{\inp{B}}, \oseq}}} \IMPLIES \subst[\sub'']{\real[r']{\seq[\fseq']{\seq[\Omega]{\inp{A \IMPLIES B}}, \lseq, \oseq}}}}
        \end{multline*}	
	Set~$\sub_i = \sub'' \comp \sub_i'$ for each $i \in \set{1, 2}$.
	By the Lifting Lemma~\ref{lem:lifting}, there exists a proof term~$t \equiv t((\sub_1 \circ r_1)(4n), (\sub_2 \circ r_2)(4n))$ such that
        \begin{multline*}
            \proves{\J}
            {
                \just{(\sub_1 \comp r_1)(4n)}{\subst[\sub_1]{\real[r_1]{\seq[\fseq']{\seq[\Omega^{\br{}}]{\out{A}}, \lseq}}}}
                \\
                \IMPLIES
                \just{(\sub_2 \comp r_2)(4n)}{\subst[\sub_2]{\real[r_2]{\seq[\fseq']{\seq[\Omega]{\inp{B}}, \oseq}}}}
                \IMPLIES
                \just{t}{\real[\sub'' \comp r']{\seq[\fseq']{\seq[\Omega]{\inp{A \IMPLIES B}}, \lseq, \oseq}}}
            }
        \end{multline*}
	Set~$r \colonequals (\sub'' \comp r') \UNI r'' \UNI \set{(4n, t)}$ which is a realisation function on $\lseq, \br[4n]{\seq[\fseq']{\seq[\Omega]{\inp{A \IMPLIES B}}, \lseq, \oseq}}$ by Lemma~\ref{lem:facts}. 
        By propositional reasoning
        \begin{multline*}
            \proves{\J}
            {
                \subst[\sub_1]{\real[r_1]{(\lseq, \br[4n]{\seq[\fseq']{\seq[\Omega^{\br{}}]{\out{A}}, \lseq}})}}
                \\
                \IMPLIES
                \subst[\sub_2]{\real[r_2]{(\lseq, \br[4n]{\seq[\fseq']{\seq[\Omega]{\inp{B}}, \oseq}})}}
                \IMPLIES
                \real{\lseq, \br[4n]{\seq[\fseq']{\seq[\Omega]{\inp{A \IMPLIES B}}, \lseq, \oseq}}}
            }
        \end{multline*}
        \qed
\end{proof}

\begin{proof}[Proof of Theorem~\ref{thm:nestreal}]
	We proceed by induction on the structure of the proof~$\pi$.
	
	For the base case, $\fseq$ is a conclusion of the~$\id$ or $\botrule$ rule.
	Apply Lemma~\ref{lem:idrule} or~\ref{lem:botrule} to construct a realisation for~$\fseq$.
	
	For the inductive case, $\fseq$ is the conclusion of a rule~$\rle$
	$$
	\vlderivation
	{
		\vliiin{\rle}{}{\fseq}
		{
			\vlhtr{\pi_1}{\fseq_1}
		}
		{\vlhy{\dots}}
		{
			\vlhtr{\pi_i}{\fseq_i}
		}
	}
	$$
	for smaller proofs~$\pi_1, \dots, \pi_i$ and premisses~$\fseq_1, \dots, \fseq_i$ where $i \in \set{1,2}$.
	By the inductive hypothesis, we have realisations $r_i$ on $\fseq_i$ such that
	$$
	\proves{\JIL}{\real[r_i]{\fseq_i}}
	$$
	Applying Lemma~\ref{lem:andrrule}, \ref{lem:orlrule}, \ref{lem:rightrules}, \ref{lem:leftrules} or~\ref{lem:implrule} corresponding to the rule~$\rle$, there exists a realisation~$r$ on~$\fseq$ and substitutions~$\sub_1, \dots, \sub_i$ (which can be $\Idnty$ if not mentioned in the Lemmas) such that
	$$
	\proves{\JIL}
	{
		\subst[\sub_1]{\real[r_1]{\fseq_1}}
		\IMPLIES
		\dots
		\IMPLIES
		\subst[\sub_i]{\real[r_i]{\fseq_i}}
		\IMPLIES
		\real{\fseq}
	}
	$$
	By the Substitution Lemma~\ref{lem:subst}
	$$
		\proves{\JIL}{\subst[\sub_j]{\real[r_j]{\fseq_j}}}
	$$
	for each $j \in \set{1, \dots, i}$.
	Using modus ponens, we achieve
		$$
	\proves{\JIL}
	{
		\real{\fseq}
	}
	$$
    \qed
\end{proof}

%\begin{corollary}
%	Let $\Log \in \set{\IK, \IKt, \IKfour, \ISfour}$.
%	%
%	Let~$A \in \langann$.
%	%
%	If
%	$$
%	\proves{\Log}{A}
%	$$
%	then there exists a realisation~$r$ on~$A$ such that
%	$$
%	\proves{\JL}{\real{A}}
%	$$
%\end{corollary}
%
%\begin{proof}
%	By Theorem~\ref{thm:soundcomp}, there is a nested sequent derivation of~$\out{A}$ and we apply Theorem~\ref{thm:nestreal} to construct a realisation.
%\end{proof}

%    \input{alternative}

\end{document}